\documentclass[final,1p]{elsarticle}
\biboptions{numbers,sort&compress}

\usepackage{pgfplots}

\mathcode`l="8000
\begingroup
\makeatletter
\lccode`\~=`\l
\DeclareMathSymbol{\lsb@l}{\mathalpha}{letters}{`l}
\lowercase{\gdef~{\ifnum\the\mathgroup=\m@ne \ell \else \lsb@l \fi}}
\endgroup

\usepackage[utf8]{inputenc}
\usepackage[T1]{fontenc}

\usepackage{csquotes}
\usepackage{amssymb,amsmath,amsfonts,amsthm}
\usepackage{array}     
\usepackage{comment}  
\usepackage{xcolor}
\usepackage{hyperref}
\usepackage{float}
\usepackage{url}
\usepackage{tikz}
\usepackage{soul}
\usepackage{siunitx}
\usepackage[utf8]{inputenc}
\usepackage{caption} 

\floatstyle{ruled}
\newfloat{algo}{tp}{lop}
\floatname{algo}{Algorithm}
\usepackage[noend]{algpseudocode}
\algrenewcommand\textproc{\sffamily}
\algrenewcommand\algorithmicindent{1em}
\usepackage{flushend}  
\usepackage{mathtools} 
\usepackage{booktabs}  

\usepackage{algorithm}
\usepackage[noend]{algpseudocode}

\def\K{\mathbb{K}}
\def\N{\mathbb{N}}
\def\Z{\mathbb{Z}}
\def\Q{\mathbb{Q}}
\def\C{\mathbb{C}}
\def\F{\mathbb{F}}
\newcommand{\M}{\mathsf{\bf M}}

\def\Id{\operatorname{Id}}

\def\inv{\operatorname{inv}}

\DeclarePairedDelimiter\floor{\lfloor}{\rfloor}

\makeatletter
\def\mproof#1{%
    \trivlist
    \item[%
        \hskip 10\p@
        \hskip \labelsep
        {\sc #1.}%
    ]
    \ignorespaces
}
\def\mqed{%
    \unskip
    \kern 10\p@
    \hfill
    \begingroup
        \unitlength\p@
        \linethickness{.4\p@}%
        \framebox(6,6){}%
    \endgroup
    \global\@qededtrue
}

\usepackage{todonotes}

\newtheorem{theorem}{Theorem}

\newtheorem{corollary}[theorem]{Corollary}

\newtheorem{definition}{Definition}
\newdefinition{example}[theorem]{Example}

\newcommand{\quotient}[2]{%
\leavevmode%
\kern-.1em
\raise.2ex\hbox{\(#1\)}\kern-.1em%
/%
\kern-.1em\lower.25ex\hbox{\(#2\)}}

\begin{document}

\title{Fast Computation of the $N$-th {T}erm of a $q$-{H}olonomic {S}equence \\ and Applications}

\author[inria]{Alin Bostan}
\ead{alin.bostan@inria.fr}

\author[univie]{Sergey Yurkevich}
\ead{sergey.yurkevich@univie.ac.at}

\address[inria]{Inria, Université Paris-Saclay (France)}
\address[univie]{University of Vienna (Austria)}

\begin{abstract}
In 1977, Strassen invented a famous baby-step/giant-step algorithm that computes the factorial $N!$ in arithmetic complexity quasi-linear in $\sqrt{N}$. In 1988, the Chudnovsky brothers generalized Strassen’s algorithm to the computation of the $N$-th term of any holonomic sequence in essentially the same arithmetic complexity. We design $q$-analogues of these algorithms. We first extend Strassen’s algorithm to the computation of the $q$-factorial of $N$, then Chudnovskys' algorithm to the computation of the $N$-th term of any $q$-holonomic sequence. Both algorithms work in arithmetic complexity quasi-linear in~$\sqrt{N}$; surprisingly, they are simpler than their analogues in the holonomic case. We provide a detailed cost analysis, in both arithmetic and bit complexity models. Moreover, we describe various algorithmic consequences, including the acceleration of polynomial and rational solving of linear $q$-differential equations, and the fast evaluation of large classes of polynomials, including a family recently considered by Nogneng and Schost.
\end{abstract}

\maketitle

%%%%%%%%%%%%%%%%%%%%%%%%%%%%%%%%%%%%%%%%%%%%%%%%%%%%%%%%%%%%%%%%%%%%
\section{Introduction}\label{sec:intro}
%%%%%%%%%%%%%%%%%%%%%%%%%%%%%%%%%%%%%%%%%%%%%%%%%%%%%%%%%%%%%%%%%%%%

A classical question in algebraic complexity theory is: how fast can one
evaluate a univariate polynomial at one point? The precise formulation of this
question depends on the model of computation. We will mainly focus on the
\emph{arithmetic complexity} model, in which one counts base field operations
at unit cost.

Horner's rule evaluates a polynomial $P$ in $O(\deg(P))$ operations.
Ostrowski~\cite{Ostrowski54} conjectured in 1954 that this is \emph{optimal for generic polynomials}, \emph{i.e.}, whose coefficients are algebraically
independent over the prime subfield. This optimality result was proved a few years later by
Pan~\cite{Pan66}. 

However, most polynomials that one might wish to evaluate 
``in practice''
have coefficients which are not algebraically independent.
Paterson and Stockmeyer~\cite{PaSc73}
showed, using the \emph{baby-step/giant-step} technique,
that for any field~$\K$, an arbitrary polynomial $P \in \K[x]$ of
degree~$N$ can be evaluated at any point in an arbitrary $\K$-algebra $A$ 
using $O(\sqrt{N})$ \emph{nonscalar}
multiplications,  \emph{i.e.}, multiplications in~$A$.
However, their algorithm uses a linear amount of scalar
multiplications, so it is not well adapted to the evaluation at points from
the base field~$\K$, since in this case the total arithmetic complexity, counted in terms of operations in $\K$, remains linear in~$N$.

 For some families of polynomials, one can do much better.
Typical examples are $x^N$ and \[P_N(x)\coloneqq x^{N-1}+ \cdots + x + 1,\] which can be
evaluated by the \emph{square-and-multiply} technique in $O(\log N)$ operations. (Note that for
$P_N(x)$ such a fast algorithm needs to perform division.) By contrast, a
family $(F_n(x))_n$ of univariate polynomials is called \emph{hard to compute}
if for large enough~$N$, the complexity of the evaluation of {$F_N$} grows at least like a power in
{$\deg(F_N)$}, whatever the algorithm used.

Paterson and Stockmeyer~\cite{PaSt71,PaSc73} proved the existence of polynomials in $\Q[x]$ which are hard to compute (note that this does not follow from Pan's result~\cite{Pan66}). 
However, their proof was based on a non-constructive argument.
Specific families of hard-to-compute
polynomials were first exhibited by Strassen~\cite{Strassen74}. 
For instance, he proved that for large $N$, the polynomial
$\sum_{\ell=0}^N 2^{2^\ell} x^\ell$ needs at least $\sqrt{N/(3 \log N)}$
operations to be evaluated.
The techniques
were refined and improved by Borodin and Cook~\cite{BoCo76},
Lipton~\cite{Lipton78} and Schnorr~\cite{Schnorr78}, who produced explicit
examples of {degree-$N$} polynomials whose evaluation requires a number of
operations linear in~$\sqrt{N}$. Subsequently, various methods have been
developed to produce similar results on \emph{lower bounds}, \emph{e.g.}, by
Heintz and Sieveking~\cite{HeSi80} using algebraic geometry, and by Aldaz et
al.~\cite{AMMP01} using a combinatorial approach. The topic is vast and very
well summarized in the book by B\"{u}rgisser, Clausen and
Shokrollahi~\cite{BCS97}.

In this article, we focus on \emph{upper bounds}, that is on the design of
fast algorithms for special families of polynomials, which are hard to
compute, but easier to evaluate than generic polynomials. For instance, for
the degree-$\binom{N}{2}$~polynomial \[Q_N(x) \coloneqq P_1(x) \cdots P_N(x),\] a
complexity in $O(N)$ is clearly achievable. We will see in~\S\ref{sec:DeFeo}
that one can do better, and attain a cost which is almost linear in~$\sqrt{N}$
(up to logarithmic factors in~$N$). Another striking example is \[R_N(x)\coloneqq\sum_{\ell=0}^N
x^{\ell^2},\] of degree~$N^2$, and whose evaluation can also be performed in
complexity quasi-linear in~$\sqrt{N}$, as shown recently by Nogneng and
Schost~\cite{NognengSchost18} (see~\S\ref{sec:NoSc}). In both cases, these
complexities are obtained by clever although somehow ad-hoc algorithms. The
starting point of our work was the question whether these algorithms for
$Q_N(x)$ and $R_N(x)$ could be treated in a unified way, which would allow
to evaluate other families of polynomials in a similar complexity.

\smallskip The answer to this question turns out to be positive. The key idea, very
simple and natural, is to view both examples as particular cases of the
following general question: 

\begin{quote} Given a $q$-holonomic sequence, that is, a
sequence satisfying a linear recurrence with polynomial coefficients in $q$
and~$q^n$, how fast can one compute its $N$-th term?
\end{quote}

In the more classical case of holonomic sequences (satisfying {linear}
recurrences with polynomial coefficients in the index~$n$), fast algorithms
exist for the computation of the $N$-th term. They rely on a basic block,
which is the computation of the factorial term~$N!$ in arithmetic complexity
quasi-linear in~$\sqrt{N}$, using an algorithm due to
Strassen~\cite{Strassen77}. The Chudnovsky brothers extended
in~\cite{ChuChu88} Strassen's algorithm to the computation of the $N$-th term
of any holonomic sequence in arithmetic complexity quasi-linear in~$\sqrt{N}$.

Our main contribution in this article consists in transferring these results to the
$q$-holonomic framework. It turns out that the resulting algorithms are
actually simpler in the $q$-holonomic case than in the usual holonomic
setting, essentially because multipoint evaluation on arithmetic progressions
used as a subroutine in Strassen's and Chudnovskys' algorithms is replaced by
multipoint evaluation on geometric progressions, which is considerably
simpler~\cite{BoSc05}. 

A consequence of our results is that the following apparently unrelated
polynomials and rational functions can be evaluated fast (note the change in
notation, with {the} variable $x$ denoted now by~$q$):

\begin{itemize}

\item $A_n(q)$, the generating function of the number of partitions into $n$
positive integers each occurring \emph{at most twice}~\cite{Yang90},
\emph{i.e.}, the coefficient of ${t}^n$ in the product 
\[\prod_{k\geq 1}(1+q^k
{t}+q^{2k}{t}^2).\]

\item $B_n(q)\coloneqq\prod_{i=1}^\infty (1-q^i) \bmod q^n$; by Euler's pentagonal
theorem~\cite[\S5]{Pak06}, \[B_n(q)=\displaystyle{1+\!\! \sum _{i(3i+1)<2n}
(-1)^i \left({q}^{\frac{i \left( 3\,i- 1 \right)}{2} }+ {q}^{\frac{i \left(
3\,i+1 \right)}{2} } \right)}.\]

\item The number $C_{n}(q)$ of $2n \times 2n$ upper-triangular matrices
over~$\F_q$ (the finite field with $q$ elements), whose square is the zero
matrix~\cite{KiMe97}; by \cite{EZ96}, $C_{n}(q)$ is equal to
\[
C_{n}(q)=\sum_{j} \left [\binom{2n}{n-3j}-  
             \binom{2n}{n-3j-1}  \right] \cdot q^{n^2-3j^2-j} .
\]

\end{itemize}

The common feature, exploited by the new algorithm, is that the sequences
$(A_n(q))_{n \geq 0}$, $(B_n(q))_{n \geq 0}$, $(C_n(q))_{n \geq 0}$ are all
$q$-holonomic. Actually, $q$-holonomic sequences are ubiquitous, so the range
of application of our results is quite broad. This stems from the fact that
they are coefficient sequences of power series satisfying $q$-difference
equations, or equivalently, $q$-shift (or, $q$-differential) equations. From
that perspective, our topic becomes intimately connected with $q$-calculus.
The roots of $q$-calculus are in works of famous mathematicians such as Rothe \cite{Rothe1793},
Gauss \cite{Gauss1808} and Heine \cite{Heine1847}. The topic gained renewed interest in the first half of the
20th century, with the work, both on the formal and analytic aspects, of
Tanner~\cite{Tanner1895},
Jackson~\cite{Jackson1909,Jackson1910b,Jackson1910},
Carmichael~\cite{Carmichael1912}, Mason~\cite{Mason1915},
Adams~\cite{Adams1928,Adams1931}, Trjitzinsky~\cite{Trjitzinsky1933},
Le~Caine~\cite{LeCaine1943} and Hahn~\cite{Hahn50}, to
name just a few. Modern accounts of the various aspects of the theory
(including historical ones) can be found in~\cite{DRSZ03,KoKeSw10,Ernst12}.

One of the reasons for interest in $q$-differential equations is that, formally,
as $q$ tends to $1$, the $q$-derivative $\frac{f(qx)-f(x)}{(q-1)x}$ tends to
$f'(x)$, thus to every differential equation corresponds a $q$-differential
equation which goes formally to the differential equation as $q\rightarrow 1$.
In nice cases, (some of) the solutions of the $q$-difference equation go to
solutions of the associated differential equation as $q\rightarrow 1$. An
early example of such a good deformation behavior is given by the basic
hypergeometric equation of Heine~\cite{Heine1847}, see also~\cite[\S1.10]{KoKeSw10}.

In computer algebra, $q$-holonomic sequences were considered starting from the
early nineties, in the context of computer-generated proofs of identities in
the seminal paper by Wilf and Zeilberger~\cite{WZ92}, notably in Section~5
(``Generalization to $q$-sums and $q$-multisums'') and in Section~6.4
(``$q$-sums and integrals''). Creative telescoping algorithms for (proper)
$q$-hypergeometric sequences are discussed in various
references~\cite{PWZ96,BoKo99,Cartier92}; several implementations of those
algorithms are described for instance
in~\cite{PauleRiese97,Riese03,KaKo09,SprengerKoepf12}.
Algorithms for computing polynomial, rational and {$q$}-hypergeometric
solutions of {$q$}-difference equations were designed by Abramov and
collaborators~\cite{Abramov95,AbBrPe95,APP98,Khmelnov00}.
These algorithms are important for several reasons. One is that they lie at
the heart of the vast generalization by Chyzak~\cite{Chyzak98,Chyzak00} of the
Wilf and Zeilberger algorithmic theory, for the treatment of general
$q$-holonomic (not only $q$-hypergeometric) symbolic summation and integration
via creative telescoping. In that context, a multivariate notion of
$q$-holonomy is needed; the foundations of the theory were laid by
Zeilberger~\cite{Zeilberger90} and Sabbah~\cite{Sabbah93} (in the language of
D-modules), see also~\cite[\S~2.5]{Cartier92} and~\cite{GaLe16}.

The simplest non-trivial holonomic sequence is $(n!)_{n\geq0}$, whose $n$-th term
combinatorially
counts the number of permutations of $n$ objects. If instead of direct
counting, one assigns to every permutation $\pi$ its number of inversions
$\inv(\pi)$, \emph{i.e.}, the number of pairs $1\leq i<j\leq n$ with
$\pi(i)>\pi(j)$, the refined count (by size and number of inversions) is
\[[n]_q!\coloneqq(1+q) (1+{q}+q^2) \cdots (1+q+\cdots +q^{n-1}).\] This is the
$q$-analogue of $n!$; it is the simplest non-trivial $q$-holonomic sequence.

There is also a natural $q$-analog of the binomial coefficients, called the
\emph{Gaussian coefficients}, defined by \[\binom{n}{k}_q\coloneqq\frac{[n]_q!}{[k]_q!
[n-k]_q!}.\] They have many counting interpretations, \emph{e.g.,} they count
the $k$-dimensional subspaces of $\F_q^n$ (points on Grassmannians
over~$\F_q$). There are $q$-analogs to (almost) everything. To select just two more basic examples, the $q$-analog~\cite[Thm.~3.3]{Andrews1976} of the binomial
theorem is given by
\begin{equation}\label{eq:BinomThm}
\prod_{k=1}^n (1+q^{k-1}x) = \sum_{k=0}^n \binom{n}{k}_q q^{\binom{k}{2}} x^k,
\end{equation}
and the
$q$-version~\cite[Thm.~3.4]{Andrews1976} of the Chu-Vandermonde identity is
\begin{equation}\label{eq:VandChu}
 \sum_{k=0}^n q^{k^2}
\binom{m}{k} _q \binom{n}{k} _q= \binom{m+n}{n}_q.
\end{equation}

The ubiquity of $q$-holonomic sequences is manifest in plenty of fields:
partition theory~\cite{Shanks51,Andrews1976,Andrews86,Yang90,Pak06,Liu17} 
and other subfields of combinatorics~\cite{FuHo85,BousquetMelou92,Kirillov95,EZ96,KiMe97,Andrews10,Yip18}; theta functions and modular forms~\cite{Bellman61,Zagier08,LTh16,Labrande18,Garvan19}; 
special functions~\cite{Borwein88,Ismail01,KoKeSw10} 
and in particular orthogonal
polynomials~\cite{Koornwinder93}; algebraic geometry~\cite{EkGe15},
representation theory~\cite{Hua00}; knot
theory~\cite{GaLe05,GaKo12,GaKo13,GaLe16,GaLaLe18}; 
Galois
theory~\cite{Hendriks97}; number theory~\cite{Osgood1971,Vizio02,AdBeDeJo17}.

\smallskip 
\begin{quote}
The main messages of this article are that for any example of a
$q$-holonomic sequence occurring in those various fields, \emph{one can compute selected coefficients faster than by a direct algorithm} and that \emph{this fact finds a tremendous number of applications}.
\end{quote}

\medskip\noindent{\bf Complexity basics.} We estimate the arithmetic complexities of
algorithms by counting arithmetic operations $(+, -, \times, {\div})$ in the
base field~$\K$ at unit cost. We use standard complexity notation, such as
$\M(d)$ for the cost of degree-$d$ multiplication in $\K[x]$, and~$\theta$ for feasible exponents of matrix multiplication. The best currently known upper bound {is} $\theta < 2.3729$~\cite{LeGall14,AlWi20}. As usual, $O(\cdot)$ stands for the big-Oh notation and $\tilde{O}(\cdot)$ is used to hide polylogarithmic factors in the argument. Most arithmetic operations on univariate
polynomials of degree $d$ in $\K[x]$ can be performed in quasi-linear
complexity~$\tilde{O}(d)$: multiplication, shift, interpolation, gcd,
resultant, \emph{etc}. A key feature of these
results is the reduction to fast polynomial multiplication, which can be
performed in time $\M(d) = O( d \log d \log \log d)$
~\cite{Schoenhage77,CaKa91}. Finally, the arithmetic cost of multiplication of polynomial matrices of size~$n$ and degree~$d$ is denoted by $\M\M(n,d)$ and we have $\M\M(n,d) = O(n^\theta d+n^2 \M(d)) = \tilde{O}(n^\theta d)$~\cite{BoSc05}. An excellent general reference for these questions is {the book by von zur Gathen and Gerhard}~\cite{GaGe13}.  \\

A short version of this article has appeared at the ISSAC'20 conference~\cite{Bostan20}. In the present version, we included the proofs of Theorems~\ref{thm:qNth-several} and~\ref{thm:qNth-bin}, we added a new Theorem~\ref{thm:detailed_complexity} containing a detailed complexity analysis of the main algorithm (\textsf{Algorithm~3}) with respect to all parameters, and we displayed pseudo-code for the algorithms as well as figures visualizing their performance. We also elaborated on a task which was mentioned as future work in the previous version, namely the application of our methods to the computation of curvatures of $q$-difference equations, see \S\ref{sec:q-diff_curv}.\\

The structure of the article is as follows: in Section~\ref{sec:motiv} we deal with the tasks of evaluating $Q_N(x)$ and $R_N(x)$. We show that these are two instances of the same problem and provide \textsf{Algorithm~3} which solves both in $O(\M(\sqrt{N}))$ arithmetic complexity. Section~\ref{sec:main} is devoted to the main results; we prove there that \textsf{Algorithm~3} can be used for computing terms of any $q$-holonomic sequence with the same cost, and provide extensions and more insight. In the same section we also consider the bit-complexity model. We identify and elaborate on several applications for our result in Section~\ref{sec:app}. In Section~\ref{sec:exp} we report on
implementations of our algorithms, which deliver encouraging timings, and we finally describe future tasks and investigation fields in Section~\ref{sec:conclusion}.

%%%%%%%%%%%%%%%%%%%%%%%%%%%%%%%%%%%%%%%%%%%%%%%%%%%%%%%%%%%%%%%%%%%%
\section{Two motivating examples}\label{sec:motiv}
%%%%%%%%%%%%%%%%%%%%%%%%%%%%%%%%%%%%%%%%%%%%%%%%%%%%%%%%%%%%%%%%%%%%

Before presenting our main results in Section~\ref{sec:main}, we describe in
this section the approach and main ideas on two basic examples. Both examples
concern the fast evaluation of special families of univariate polynomials.
In~\S\ref{sec:DeFeo}, we consider polynomials of the form $\prod_\ell (x -
q^\ell)$, and in~\S\ref{sec:NoSc} sparse polynomials of the form $\sum_{\ell}
p^\ell x^{a \ell^2+b \ell}$. In both cases, we first present fast ad-hoc
algorithms, then introduce equally fast alternative algorithms, which have the
nice feature that they will be generalizable to a broader setting.

%%%%%%%%%%%%%%%%%%%%%%%%%%%%%%%%%%%%%%%%%%%%%%%%%%%%%%%%%%%%%%%%%%%%
\subsection{Evaluation of some structured polynomials}\label{sec:DeFeo}
%%%%%%%%%%%%%%%%%%%%%%%%%%%%%%%%%%%%%%%%%%%%%%%%%%%%%%%%%%%%%%%%%%%%

Here is our first example, that emerged from a question asked to the first author by
Luca De Feo (private email communication, 10 January 2020); this
was the starting point of the article.

\begin{quote}
    Let $q$ be an element of the field $\K$, and consider the polynomial
\begin{equation} \label{eq:defeo}
F(x) \coloneqq \prod_{i=0}^{N-1} (x - q^i) \; \in \K[x].
\end{equation}
Given another element $\alpha\in\K$, how fast can one evaluate $F(\alpha)$?
\end{quote}

{If $q=0$, then $F(\alpha)=\alpha^N$ can be computed in $O(\log N)$ operations in $\K$, by binary powering. We assume in what follows that~$q$ is nonzero.}
Obviously, a direct algorithm consists in computing the successive powers $q,
q^2, \ldots, q^{N-1}$ using $O(N)$ operations in $\K$, then computing the
elements $\alpha-1, \alpha-q, \ldots, \alpha-q^{N-1}$ in $O(N)$ more
operations in $\K$, and finally returning their product. The total arithmetic
cost of this algorithm\footnote{If $q^n = 1$ for some $n<N$, then it is enough to compute the product of $\alpha-q^i$ for $i=0,\dots,n-1$ and its appropriate power. The latter step can be done efficiently (in essentially $\log(N)$ operations) using binary powering. Our main interest lies therefore in $q\in \K$ that are not roots of unity of small order.} is $O(N)$, linear in the degree of $F$.

\medskip 
Is it possible to do better?
The answer is positive, as one can use the following \emph{baby-step/giant-step} strategy, in which, in order to simplify things,
we assume that $N$ is a perfect square, 
$N=s^2$.\\

\medskip \noindent \underline{\textsf{Algorithm 1}}\label{algo:1}
\begin{enumerate}
  \item (Baby-step) Compute the values of $q, q^2, \ldots, q^{s-1}$, and deduce
the coefficients of the polynomial \[G(x)\coloneqq\prod_{j=0}^{s-1} (x - q^j).\]
  \item (Giant-step) Compute $Q\coloneqq q^s, Q^2, \ldots, Q^{s-1}$, and deduce
the coefficients of the polynomial \[H(x)\coloneqq\prod_{k=0}^{s-1} (\alpha - Q^k \cdot x).\]
  \item Return the resultant $\textrm{Res}(G,H)$.
\end{enumerate}
By the basic property of resultants, the output of this algorithm is
\[\textrm{Res}(G,H) \! = \! \prod_{j=0}^{s-1} H(q^j) = \prod_{j=0}^{s-1}
\prod_{k=0}^{s-1} \left(\alpha - q^{sk + j} \right) = 
\prod_{i=0}^{N-1} (\alpha - q^i) =
F(\alpha).\] 
Using the fast subproduct tree algorithm~\cite[Algorithm~10.3]{GaGe13}, one
can perform the baby-step~(1) as well as the giant-step~(2) in $O(\M(\sqrt{N})
\log N)$ operations in $\K$, and by~\cite[Corollary~11.19]{GaGe13} the same
cost can be achieved for the resultant computation in step~(3). Using fast
polynomial multiplication, we conclude that $F(\alpha)$ can be computed in
arithmetic complexity quasi-linear in~$\sqrt{N}$.

{Note that if $N$ is not a perfect square, then  one can compute $F(\alpha)$ as $F(\alpha) = F_1(\alpha)  F_2(\alpha)$, where
$F_1(\alpha)\coloneqq \prod_{i=0}^{{\lfloor \sqrt{N} \rfloor}^2-1} (\alpha - q^i)$
is computed as in {\textsf{Algorithm 1}}, 
while 
$F_2(\alpha)\coloneqq \prod_{i = {\lfloor \sqrt{N} \rfloor}^2}^{N-1} (\alpha - q^i)$
can be computed naively, since $N - {\lfloor \sqrt{N} \rfloor}^2 = O(\sqrt{N})$.
}

\smallskip It is possible to speed up the previous algorithm by a logarithmic factor
in~$N$ using a slightly different scheme, still based on a \emph{baby-step/giant-step} strategy, but exploiting the fact that the roots of~$F$ are in
geometric progression. Again, we assume that $N=s^2$ is a perfect square.
This alternative algorithm goes as follows.
Note that it is very close in spirit to Pollard's 
algorithm described on page 523 of~\cite{Pollard}.

\bigskip \noindent \underline{\textsf{Algorithm 2}}\label{algo:2}
\begin{enumerate}
  \item (Baby-step) Compute $q, q^2, \ldots, q^{s-1}$, and deduce
the coefficients of the polynomial $P(x)\coloneqq\prod_{j=0}^{s-1} (\alpha - q^j \cdot x)$.
  \item (Giant-step) 
Compute $Q\coloneqq q^s, Q^2, \ldots, Q^{s-1}$, and 
evaluate $P$ {simultaneously} at $1, Q, \ldots, Q^{s-1}$.
  \item Return the product $ P(Q^{s-1}) \cdots P(Q) P(1)$.
\end{enumerate}
Obviously, the output of this algorithm is 
% \vspace{-0.15cm}
\[ \prod_{k=0}^{s-1} P(Q^k) = 
\prod_{k=0}^{s-1} \prod_{j=0}^{s-1} (\alpha - q^j \cdot q^{sk}) = 
\prod_{i=0}^{N-1} (\alpha - q^i) =
F(\alpha).
\]
As pointed out in the {remarks} after the proof of~\cite[Lemma~1]{BoSc05}, one
can compute $P(x)=P_s(x) = \prod_{j=0}^{s-1} (\alpha - q^j \cdot x)$ in
step~(1) without computing the subproduct tree, by using a divide-and-conquer
scheme which exploits the fact that $P_{2t}(x) = P_t(q^t x) \cdot P_t(x)$
and $P_{2t+1}(x) =   (\alpha-q^{2t}x) \cdot P_t(q^t x) \cdot P_t(x)$. The cost
of this algorithm is $O(\M(\sqrt{N}))$ operations in $\K$.

As for step~(2), one can use the fast \emph{chirp transform} algorithms of Rabiner,
Schafer and Rader~\cite{RaScRa69} and of Bluestein~\cite{Bluestein70}. These
algorithms rely on the following observation: writing $Q^{ij} =
Q^{\binom{i+j}{2}} \cdot Q^{-\binom{i}{2}} \cdot Q^{-\binom{j}{2}}$ and
$P(x)=\sum_{j=0}^{{s}} c_j x^j$ implies that the needed values ${P(Q^i) =
\sum_{j=0}^{{s}} c_j Q ^{ij}, 0\leq i < s}$, are
\[P(Q^i) = Q^{-\binom{i}{2}} \cdot {\sum_{j=0}^{{s}} 
c_j Q^{-\binom{j}{2}} \cdot Q^{\binom{i+j}{2}}}, \;  0\leq i < s,
\]
in which the sum is simply the coefficient of $x^{{s}+i}$ in the product
\[{ \left( \sum_{{j}=0}^{{s}} c_{{j}} Q^{-\binom{{j}}{2}}  x^{{s}-{j}} \right) \left( 
\sum_{\ell=0}^{{2s}}  Q^{\binom{\ell}{2}} x^{\ell} \right) }.\]
This polynomial product can be computed in $2 \, \M(s)$ operations (and even in
$\M(s) + O(s)$ using the \emph{transposition principle}~\cite{HQZ04,BoLeSc03}, since
only the median coefficients $x^{{s}}, \ldots, x^{{2s-1}}$ are actually needed).
In conclusion, step~(2) can also be performed in $O(\M(\sqrt{N}))$ operations
in $\K$, and thus $O(\M(\sqrt{N}))$ is the total cost of this second
algorithm.

\smallskip We have chosen to detail this second algorithm for several reasons: not only
because it is faster by a factor $\log(N)$ compared to the first one, but more
importantly because it has a {simpler} structure, which will be
generalizable to {the} general $q$-holonomic setting. In fact, we do not provide a pseudo-code implementation for this algorithm, since we will do so for the more general case (Algorithm~\ref{alg:algo3step2}).

%%%%%%%%%%%%%%%%%%%%%%%%%%%%%%%%%%%%%%%%%%%%%%%%%%%%%%%%%%%%%%%%%%%%
\subsection{Evaluation of some sparse polynomials} \label{sec:NoSc}
%%%%%%%%%%%%%%%%%%%%%%%%%%%%%%%%%%%%%%%%%%%%%%%%%%%%%%%%%%%%%%%%%%%%

Let us now consider the sequence of sparse polynomial sums
\[
v^{(p,a,b)}_N(q) = \sum_{n=0}^{N-1} p^n q^{an^2+bn},
\]
where $p \in \K$ and $a,b\in\mathbb{Q}$ such that $2a, a+b$ are both integers.
Typical examples are (truncated) modular forms~\cite{PaRa18}, which are ubiquitous in complex analysis~\cite{Bellman61},
number theory~\cite{Zagier08} and  combinatorics~\cite{Andrews1976}.
For instance, the \emph{Jacobi theta function} $\vartheta_3$  
depends on two complex variables $z\in\C$, and $\tau\in\C$ with $\Im 
(\tau)>0$, and it is defined by
\begin{equation*}
\vartheta_3(z; \tau) 
 = \sum_{n=-\infty}^\infty e^{\pi i (n^2 \tau + 2 n z)} 
= 1 + 2 \sum_{n=1}^\infty \eta^n q^{n^2}, 
\end{equation*}
where $q = e^{\pi i \tau}$ is the nome ($|q|<1$) and $\eta = e^{2\pi i z}$.
Here, $\K=\C$.
Another example is the \emph{Dedekind eta function}, appearing in Euler's famous \emph{pentagonal theorem}~\cite[\S5]{Pak06}, which has a similar form
\[
q^{\frac{1}{24}} \cdot \left( 1+\sum _{n=1}^{\infty} (-1)^n
\left(
{q}^{\frac{n \left( 3\,n-
1 \right)}{2} }+  {q}^{\frac{n \left( 3\,n+1 \right)}{2} } \right) \right),
\quad \text{with} \; q = e^{2 \pi i \tau}.
\]
Moreover, sums of the form $v^{(1,a,b)}_N(q) = \sum_{n=0}^{N-1} q^{a n^2+bn}$,
over $\K={\Q}$ or $\K=\F_2$,
crucially occur in a recent algorithm by Tao, Crott and Helfgott~\cite{TCH12}
for the efficient construction of prime numbers in given intervals, \emph{e.g.}, in
the context of effective versions of Bertrand's postulate.
Actually, (the proof of) Lemma 3.1 in~\cite{TCH12} contains the first
sublinear complexity result for the evaluation of the sum
$v^{({p},a,b)}_N(q)$ at an arbitrary point~$q$; namely, the cost is
$O(N^{\theta/3})$, where $\theta \in [2,3]$ is any feasible exponent for matrix
multiplication. Subsequently, Nogneng and Schost~\cite{NognengSchost18}
designed a faster algorithm, and lowered the cost down to
$\tilde{O}(\sqrt{N})$. Our algorithm is similar in spirit to theirs, as it
also relies on a \emph{baby-step/giant-step} strategy.

Let us first recall the principle of the Nogneng-Schost
algorithm~\cite{NognengSchost18}. 
Assume as before that $N$ is a perfect square, $N=s^2$.
The starting point is the remark that
\[
v^{(p,a,b)}_N(q) 
= \sum_{{n}=0}^{N-1} p^{{n}} q^{a {n}^2+b {n}}
= \sum_{k=0}^{s-1}\sum_{j=0}^{s-1} p^{j+sk} q^{a (j+sk)^2+b(j+sk)}
\]
can be written 
\[
\sum_{k=0}^{s-1}p^{sk} q^{a s^2k^2+bsk} \cdot P(q^{2ask}), \;
\text{where} \; P(y) \coloneqq \sum_{j=0}^{s-1} p^j q^{aj^2+bj} y^j.
\]
Therefore, the computation of $v^{(p,a,b)}_N(q)$ can be reduced essentially to
the {simultaneous} evaluation of the
polynomial $P$ at $s=1+\deg(P)$
points (in geometric progression), with arithmetic cost $O(\M(\sqrt{N}))$.

We now describe an alternative algorithm, of {similar complexity}
$O(\M(\sqrt{N}))$, with a {slightly} larger constant in the big-Oh
estimate, but whose advantage is its potential of generality.

Let us denote by $u_n(q)$ the summand $p^n q^{an^2+bn}$. Clearly, the sequence 
$\left(u_n(q) \right)_{n \geq 0}$ 
satisfies the recurrence relation
\[
u_{n+1}(q) = A(q,q^n) \cdot u_n(q), \quad \text{where} \quad  A(x,y) \coloneqq p x^{a+b} y^{2a}.
\]
As an immediate consequence, the sequence with general term  $v_n(q) 
\coloneqq \sum_{k=0}^{n-1} u_k(q)$ 
satisfies a similar recurrence relation
\begin{equation}\label{eq:v}
v_{n+2}(q) - v_{n+1}(q) = A(q,q^n) \cdot (v_{n+1}(q) - v_{n}(q)),
\end{equation}
with initial conditions 
$v_0(q) = 0$ and  $v_1(q) = 1$.
This scalar recurrence of order two is equivalent to the first-order matrix recurrence
\[
\begin{bmatrix}
v_{n+2} \\ v_{n+1}
\end{bmatrix}	
=
\begin{bmatrix}
A(q,q^n) + 1 & - A(q,q^n)\\ 1 & 0
\end{bmatrix}
\times 
\begin{bmatrix}
v_{n+1} \\ v_{n}
\end{bmatrix}.
\]
By unrolling this matrix recurrence, we deduce that
\[
\begin{bmatrix}
v_{n+1} \\ v_{n}
\end{bmatrix}	
=
M(q^{n-1}) 
\begin{bmatrix}
v_{n} \\ v_{n-1}
\end{bmatrix}
=
M(q^{n-1}) \cdots M(q) M(1)
\times
\begin{bmatrix}
1 \\ 0
\end{bmatrix},
\]
where
\[
M(x)
\coloneqq
\begin{bmatrix}
p q^{a+b} x^{2a}+ 1 & \! \! - p q^{a+b} x^{2a}\\ 1 & \!\! 0
\end{bmatrix}
,	
\]
{hence $ v_N = 
\begin{bmatrix}
0 & 1
\end{bmatrix}	
\times
M(q^{N-1}) \cdots M(q) M(1)
\times
\begin{bmatrix}
1 \\ 0
\end{bmatrix}.
$}
Therefore, the computation of $v_N$ reduces to the computation of
{the} ``matrix $q$-factorial'' $M(q^{N-1}) \cdots M(q) M(1)$, which can be
performed fast by using a \emph{baby-step/giant-step} strategy similar to
the one of the second algorithm in~\S\ref{sec:DeFeo}.
Again, we assume for simplicity that $N=s^2$ is a perfect square.
The algorithm goes as follows.

\bigskip \noindent \underline{\textsf{Algorithm 3}}
\, {(matrix $q$-factorial)} \label{algo:3}
\begin{enumerate}
  \item[(1)] (Baby-step) Compute $q, q^2, \ldots, q^{s-1}$; deduce
the coefficients of the {polynomial matrix} $P(x)\coloneqq M(q^{s-1} x) \cdots M(q x) M(x)$.
  \item[(2)] (Giant-step) 
Compute $Q\coloneqq q^s, Q^2, \ldots, Q^{s-1}$, and 
evaluate (the entries of) $P(x)$ {simultaneously} at $1, Q, \ldots, Q^{s-1}$.
  \item[(3)] Return the product $P(Q^{s-1}) \cdots P(Q) P(1)$.
\end{enumerate}
Clearly, this algorithm generalizes \textsf{Algorithm~2} in~\S\ref{sec:DeFeo} and, as promised, we also provide a detailed pseudo-code implementation: \hyperref[alg:algo3step1]{{\sf Step1}} and \hyperref[alg:algo3step2]{{\sf Step2 \& Step3}}:
\setcounter{algorithm}{2}
\begin{algorithm}[H]
\begin{algorithmic}[1]
\State $q_s \gets [q,q^2,\dots,q^{s-1}]$ 
\State $t \gets s$ 
\Function{$\mathcal{BS}$}{$t$}
\If{$t=1$}
\State \Return $M(x)$
\EndIf
\If{$t$ is even}
\State $p_1(x) \gets \mathcal{BS}(t/2)$
\State $p_2(x) \gets p_1(q^{t/2}x)$
\Comment{Using $q_s$}
\State \Return $p_2(x) \cdot p_1(x)$
\Comment{Fast polynomial multiplication}
\Else
\State $p_1(x) \gets \mathcal{BS}((t-1)/2)$
\State $p_2(x) \gets p_1(q^{(t-1)/2}x)$
\Comment{Using $q_s$}
\State $p_3(x) \gets M(q^{t-1} x)$
\Comment{Using $q_s$}
\State \Return $p_3(x) \cdot p_2(x) \cdot p_1(x)$
\Comment{Fast polynomial multiplication}
\EndIf
\EndFunction
\end{algorithmic}
\caption{(\textcolor{red}{\textsf{Step1}}) \quad \\ {\bf Input}: $s,q,M(x)$ \qquad \\ {\bf Output}: $M(q^{s-1}x)\cdots M(qx)M(x)$}
\label{alg:algo3step1}
\end{algorithm}  
\setcounter{algorithm}{2}
\begin{algorithm}[H]
\textbf{Assumptions:} $Q \neq 0$, $P(x)$ polynomial matrix of size $n \times n$ and degree $d \geq s$.
\begin{algorithmic}[1]
\State $Q_d \gets [Q,Q^2,\dots,Q^{d-1}]$ 
\State $Q' \gets 1/Q$
\State $Q'_d \gets [Q^{-\binom{d}{2}},\dots, Q^{-\binom{1}{2}},Q^{\binom{0}{2}},\dots,Q^{\binom{2d}{2}}]$
\Comment{Using $Q_d$ and $Q'$}
\State $P_{s-1},\dots, P_0$ 
\Comment{Empty $n\times n$ matrices}
\For{$i$ from 1 to $n$}
\For{$j$ from 1 to $n$}
\State $p(x) \gets P(x)_{i,j}$  
\Comment{$p(x) = c_0+c_1x+\cdots+c_dx^d$}
\State $p_1(x) \gets \sum_{\ell=0}^d c_{\ell}  Q^{-\binom{\ell}{2}} x^{d-\ell}$
\Comment{Using $Q'_d$}
\State $p_2(x) \gets \sum_{\ell=0}^{2d} Q^{\binom{\ell}{2}} x^{\ell}$
\Comment{Using $Q'_d$}
\State  $p_3(x) \gets P_1(x)\cdot P_2(x)$
\Comment{$p_3(x) = \sum_{\ell=0}^{3d} r_\ell x^\ell$; fast multiplication}
\For{$k$ from 0 to $s-1$}
\State $(P_k)_{i,j} \gets r_{d+k}$ 
\Comment{$P_\ell = P(Q^\ell)$ for $\ell = 0,\dots,s-1$}
\EndFor
\EndFor
\EndFor
\State $P \gets 1$
\For{$k$ from $0$ to $s-1$}
    \State $P \gets P_k \cdot P$
\EndFor
\Return $P$
\end{algorithmic}
\caption{(\textcolor{red}{\textsf{Step2 \& Step3}}) \quad \\ {\bf Input}: $s, Q, P(x)$ \qquad \\ {\bf Output}: $P(Q^{s-1})\cdots P(1)$}
\label{alg:algo3step2}
\end{algorithm}  

By the same observations as in \textsf{Algorithm~2} in~\S\ref{sec:DeFeo}, {the complexity of} \textsf{Algorithm 3}
already is quasi-linear in~$\sqrt{N}$. In the next section we will discuss the complexity not only with respect to $N$, but to the matrix size and degree as well.

We remark that when applied to the computation of $v^{(p,a,b)}_N(q)$, the dependence
in~$a,b$ of \textsf{Algorithm~3} is quite high (quasi-linear in~$a$ and~$b$). If $a$ and $b$ are fixed
and considered {as}~$O(1)$ this dependence is invisible, but otherwise the
following variant has the same complexity with respect to~$N$, {and} a much
better cost with respect to~$a$ and~$b$. It is based on the simple observation
that, if $\tilde{M}(x)$ denotes the polynomial matrix
\begin{equation}\label{eq:tildeM}
\tilde{M}(x)
\coloneqq
\begin{bmatrix}
p r x + 1 & \! \! - p r x\\ 1 & \!\! 0
\end{bmatrix},
\;
\text{with}
\;
r \coloneqq q^{a+b},
\end{equation}
and if $\tilde{q} \coloneqq q^{2a}$, then the following matrix $q$-factorials coincide:
\[
M(q^{N-1}) \cdots M(q) M(1)
=
\tilde{M}(\tilde{q}^{N-1}) \cdots {\tilde{M}}(\tilde{q}) {\tilde{M}}(1)
.\]

\bigskip \noindent \underline{\textsf{Algorithm 4}}
\, {(matrix $q$-factorial, variant)} \label{algo:4}
\begin{enumerate}
  \item[(0)] (Precomputation) Compute $r\coloneqq q^{a+b}$, $\tilde{q}\coloneqq q^{2a}$, and $\tilde{M}$ in~\eqref{eq:tildeM}.
  \item[(1)] (Baby-step) Compute $\tilde{q}, \tilde{q}^2, \ldots, \tilde{q}^{s-1}$; deduce the coefficients of the {polynomial matrix}  \[\tilde{P}(x)\coloneqq\tilde{M}(\tilde{q}^{s-1} x) \cdots \tilde{M}(\tilde{q} x) \tilde{M}(x).\]
  \item[(2)] (Giant-step) 
Compute $\tilde{Q}\coloneqq\tilde{q}^s, \tilde{Q}^2, \ldots, \tilde{Q}^{s-1}$, and 
evaluate (the entries of) $\tilde{P}(x)$ {simultaneously} at $1, \tilde{Q}, \ldots, \tilde{Q}^{s-1}$.
  \item[(3)] Return the product $\tilde{P}(\tilde{Q}^{s-1}) \cdots \tilde{P}(\tilde{Q}) \tilde{P}(1)$.
\end{enumerate}

Using binary powering, the cost of the additional precomputation in step~(0)
is only logarithmic in $a$ and $b$. In exchange, the new steps~(2) and~(3) are
performed on matrices whose degrees do not depend on $a$ and~$b$ anymore (in
the previous, unoptimized, version the degrees of the polynomial matrices were linear in~$a$ and~$b$). The total arithmetic cost with respect to~$N$ is still quasi-linear in~$\sqrt{N}$.\\

In the next section, we will show that \textsf{Algorithm~3} can be employed for the fast computation of the $N$-th term of \emph{any} $q$-holonomic sequence. Note that the trick in \textsf{Algorithm~4} relies on the fact that $M(x)$, coming from the recurrence for $v_N^{(p,a,b)}(q)$, contains only pure powers of $x$ and $q$. We cannot hope for this phenomenon in general, however we advise to bear this simplification in mind for some practical purposes. In any case, we can improve on the quasi-linear cost in the degree $d$ of the polynomial matrix $M(x)$ in \textsf{Algorithm~3}, obtaining a complexity of essentially $\sqrt{d}$; in essence, the idea consists in choosing $s=\sqrt{N/d}$ rather than $\sqrt{N}$, see \S\ref{sec:qNth_complexity}.

%%%%%%%%%%%%%%%%%%%%%%%%%%%%%%%%%%%%%%%%%%%%%%%%%%%%%%%%%%%%%%%%%%%%
\section{Main results}\label{sec:main}
%%%%%%%%%%%%%%%%%%%%%%%%%%%%%%%%%%%%%%%%%%%%%%%%%%%%%%%%%%%%%%%%%%%%

In this section, we generalize the algorithms from~\S\ref{sec:motiv},
and show that they apply to the general setting of $q$-holonomic sequences.

\subsection{Preliminaries}\label{sec:prelim}

A sequence $u_n = u_n(q)$ is $q$-holonomic if it satisfies a nontrivial $q$-recurrence, that is, a linear recurrence with coefficients given by polynomials in $q$ and~$q^n$. 

\begin{definition}[$q$-holonomic sequence] \label{def:qhol}
Let $\K$ be a field, and $q\in\K$.
A sequence $(u_n(q))_{n \geq 0}$ in $\K^{\N}$
is called \emph{$q$-holonomic} if there exist
$r\in \N$ and polynomials $c_0(x,y), \ldots, c_r(x,y)$ in $\K[x,y]$, with $c_r(x,y)\neq 0$, such
that
\begin{equation}\label{eq:qrec}
c_r(q, q^n) u_{n + r}(q) + \cdots  + c_0(q, q^n) u_n(q) = 0, 
\quad \text{for  all } \; n\geq 0.
\end{equation}
The integer $r$ is called the \emph{order} of the
$q$-recurrence~\eqref{eq:qrec}. When $r=1$, we say that $(u_n(q))_{n \geq 0}$
is $q$-hypergeometric. \end{definition}

The most basic examples are the \emph{$q$-bracket} and the 
\emph{$q$-factorial}, 
\begin{equation}\label{def:qbracket}
{[n]_q \coloneqq 1 + q + \cdots + q^{n-1}}
\quad \text{and}\quad [n]_q! \coloneqq \prod_{k=1}^n [k]_q.
\end{equation}	
They are clearly $q$-holonomic, and even $q$-hypergeometric.

The sequences $(u_n)_{n \geq 0} = (q^n)_{n \geq 0}$, $(v_n)_{n \geq 0} = (q^{n^2})_{n \geq 0}$ and  $(w_n)_{n \geq 0} = (q^{\binom{n}{2}})_{n \geq 0}$ are also $q$-hypergeometric,
since they satisfy the recurrence relations
\[ u_{n+1} - q u_n = 0, \quad v_{n+1}- q^{2n+1} v_n=0, \quad w_{n+1}-q^n w_n=0.\]
However, the sequence $(q^{n^3})_{n \geq 0}$ is not $q$-holonomic~\cite[Ex.~2.2(b)]{GaLe16}.
More generally, this also holds for the sequence $(q^{n^s})_{n \geq 0}$, for any $s>2$, see~\cite[Th.~4.1]{Bezivin92} and also \cite[Th.~1.1]{Garoufalidis11}.

Another basic example is the $q$-Pochhammer symbol  
\begin{equation}\label{def:qPochhammer}
(x;q)_n \coloneqq \prod_{k=0}^{n-1} (1- x q^{k}),
\end{equation}
which is also $q$-hypergeometric, since
$(x;q)_{n+1} - (1-xq^{n}) (x;q)_n = 0.$
In particular, the sequence
 $(q;q)_n \coloneqq \prod_{k=1}^{n} (1-q^k)$, also denoted $(q)_n$,
is $q$-hypergeometric and satisfies 
$(q)_{n+1} - (1-q^{n+1}) (q)_n = 0.$ In Section~\S\ref{sec:motiv} we encountered $v_n^{(p,a,b)} = \sum_{k=0}^n p^k q^{ak^2+bk}$, which is $q$-holonomic (see Eq.~(\ref{eq:v})), but generally not $q$-hypergeometric. 

Note that (\ref{eq:qrec}) reduces to a C-linear recurrence, \emph{i.e.} a linear recurrence with \emph{constant} coefficients, if all polynomials $c_0(x,y),\dots,c_r(x,y)$ are constant in the variable $y$. For these kinds of sequences there exist quasi-optimal algorithms~\cite{MiBr66,Fiduccia85,BoMo21}, therefore we assume from now on that the maximal degree $d$ of $c_0(x,y),\dots,c_r(x,y)$ in $y$ is positive.

As mentioned in the introduction, $q$-holonomic sequences show up in various
contexts. As an example, in (quantum) knot theory, the (``colored'') Jones
function of a (framed oriented) knot (in 3-space) is a powerful knot
invariant, related to the Alexander polynomial~\cite{BaGa96}; it is a
$q$-holonomic sequence of Laurent polynomials~\cite{GaLe05}. Its recurrence
equations are themselves of interest, as they {are} closely related 
to the A-polynomial of a knot, via the \emph{AJ conjecture}~\cite{Garoufalidis04,Garoufalidis18,DeGa20},
verified in some cases using massive computer algebra calculations~\cite{GaKo13}.

It is well known that the class of $q$-holonomic sequences is closed under
several operations, such as addition, multiplication, 
{Hadamard product and monomial substitution}~\cite{KPS07,KaKo09,GaLe16}.
All these closure properties
are effective, \emph{i.e.}, they can be executed algorithmically on the level of
$q$-recurrences. Several computer algebra packages are available for the
manipulation of $q$-holonomic sequences, 
\emph{e.g.}, the Mathematica packages
\textcolor{magenta}{\href{https://www3.risc.jku.at/research/combinat/software/ergosum/RISC/qGeneratingFunctions.html}{
\textsf{qGeneratingFunctions}}}~\cite{KaKo09}
{and}
\textcolor{magenta}{\href{https://www3.risc.jku.at/research/combinat/software/ergosum/RISC/HolonomicFunctions.html}{\textsf{HolonomicFunctions}}}~\cite{Koutschan10}, and the Maple packages
\textcolor{magenta}{\href{http://www.hypergeometric-summation.org}{
\textsf{qsum}}}~\cite{BoKo99},
\textcolor{magenta}{\href{http://www.hypergeometric-summation.org}{
\textsf{qFPS}}}~\cite{SprengerKoepf12},
\textcolor{magenta}{\href{https://qseries.org/fgarvan/qmaple/qseries/index.html}{\textsf{qseries}}}
and
\textcolor{magenta}{\href{https://maplesoft.com/support/help/Maple/view.aspx?path=QDifferenceEquations}{\textsf{QDifferenceEquations}}}.

\smallskip A simple but useful fact is that the order-$r$ scalar
$q$-recurrence~\eqref{eq:qrec} can be translated into a
first-order recurrence on {$r\times 1$ vectors}:
\begin{equation}\label{eq:qrec-mat}
\begin{bmatrix}
u_{n+r} \\ \vdots \\u_{n+1}
\end{bmatrix}	
=
\begin{bmatrix}
 -\frac{c_{r-1}}{c_r} & \cdots &  -\frac{c_1}{c_r} &  -\frac{c_0}{c_r}
\\ 1 & \cdots & 0 & 0
\\ \vdots  & \ddots & \vdots & \vdots & 
\\ 0 & \cdots & 1 & 0 
\end{bmatrix}
\times 
\begin{bmatrix}
u_{n+r-1} \\ \vdots \\ u_{n}
\end{bmatrix}.
\end{equation}

In particular, the $N$-th term of {the} $q$-holonomic sequence {$(u_n)$} is simply expressible in terms of the \emph{matrix $q$-factorial}
\begin{equation}\label{eq:MatQfact}
M(q^{N-1}) \cdots M(q)M(1),
\end{equation} where
$M(q^n)$ denotes the companion matrix from equation~\eqref{eq:qrec-mat}. This observation is crucial, since it exposes the connection to the algorithms presented in the previous section.

%%%%%%%%%%%%%%%%%%%%%%%%%%%%%%%%%%%%%%%%%%%%%%%%%%%%%%%%%%%%%%%%%%%

\subsection{Computation of the $q$-factorial}\label{sec:qfact}

We now give the promised $q$-analogue of Strassen's result on the
computation of $N!$ {in $O(\M(\sqrt{N}) \log N)$ arithmetic operations}.
{Note that Strassen's case $q=1$ is also covered by~\cite[\S6]{BoGaSc07}, 
where the cost $O(\M(\sqrt{N}))$ is reached
under some invertibility assumptions.}

\begin{theorem} \label{thm:qfact}
{Let $\K$ be a field,} 
let $q\in\K\setminus \{1\}$ and $N\in \N$.
The $q$-factorial $[N]_q!$ can be computed using $O(\M(\sqrt{N}))$ operations in~$\K$. The same is true for the $q$-Pochhammer symbol $(\alpha;q)_N$ for any $\alpha\in\K$.
\end{theorem}

\begin{proof}
{If $\alpha=0$, then $(\alpha;q)_N=1$}.
If $q=0$, then $[N]_q!=1$ {and $(\alpha;q)_N=1-\alpha$}.
We can assume that 
{$q\in\K\setminus \{0, 1\}$} 
{and $\alpha\in\K\setminus \{0 \}$}.
We have {$[N]_q!=  
r^{N} \cdot F(q^{-1})$} 
and {$(\alpha;q)_N =  \alpha^N \cdot F(\alpha^{-1})$},
{where $r\coloneqq q/(1-q)$ and $F(x)\coloneqq\prod_{i=0}^{N-1} (x-q^i)$}. \textsf{Algorithm 2} can be used to compute 
{$F(q^{-1})$ and $F(\alpha^{-1})$}
in $O(\M(\sqrt{N}))$ operations in~$\K$.
The cost of computing $r^{N}$ and $\alpha^{N}$ is $O(\log N)$, 
and thus it is negligible.
\end{proof}

\begin{corollary} \label{coro:qbin}
Under the assumptions of Theorem~\ref{thm:qfact} and for any $n\in\N$, one can compute in $O(\M(\sqrt{N}))$ operations in~$\K$:
\begin{itemize}
\item the $q$-binomial coefficient $\binom{N}{n}_q$;
\item the coefficient of $x^n$ in the polynomial $\prod_{k=1}^N (1+q^{k-1}x)${;}
\item {the sum $\binom{{N-n}}{0}_q \binom{{n}}{0}_q  + q \binom{{N-n}}{1}_q \binom{{n}}{1}_q + \cdots + q^{{n}^2}\binom{{N-n}}{{n}}_q \binom{{n}}{{n}}_q$.}
\end{itemize}
\end{corollary}

\begin{proof}
The first assertion is a direct consequence of Theorem~\ref{thm:qfact}.
The second assertion is a consequence of the first one, and of~\eqref{eq:BinomThm}.
The third assertion is a consequence of the first one, and of~\eqref{eq:VandChu}.
\end{proof}
	
\subsection{$N$-th term of a $q$-holonomic sequence}\label{sec:qNth}

We now  offer the promised $q$-analogue of Chudnovskys' result on
the computation of the $N$-th term of an arbitrary holonomic sequence
{in $O(\M(\sqrt{N}) \log N)$ arithmetic operations}.
{Note that Chudnovskys' case $q=1$ is also covered by~\cite[\S6]{BoGaSc07}, 
where the improved cost $O(\M(\sqrt{N}))$ is reached
under additional invertibility assumptions.}

\begin{theorem} \label{thm:qNth}
Let $\K$ be a field, 
$q\in \K\setminus \{ 1 \}$ and $N\in \N$.
Let  $(u_n(q))_{n \geq 0}$ be a $q$-holonomic sequence satisfying recurrence~\eqref{eq:qrec}, and assume that $c_r(q,q^k)$ is nonzero for $k=0,\ldots, N-1$. Then, 
$u_N(q)$ can be computed in $O(\M(\sqrt{N}))$ operations in~$\K$.
\end{theorem}

\begin{proof}
Using equation~\eqref{eq:qrec-mat}, it is enough to show that the matrix
$q$-factorial 
\[
M(q^{N-1}) \cdots M(q)M(1)
\]
can be computed in $O(\M(\sqrt{N}))$,
where $M(q^n)$ denotes the companion matrix from equation~\eqref{eq:qrec-mat}.
{\sf Algorithm 3} adapts \emph{mutatis mutandis} to this effect.
\end{proof}

Remark that if $q$ is a root of unity of order $n<N$, then the computation of $U_N(q) = M(q^{N-1}) \cdots M(q)M(1)$ can be simplified using 
\[
U_N(q) = M(q^{k})\cdots M(1) \cdot U_n(q)^r,
\]
where $r=\floor{(N-1)/n}$ and $k=N-1 - rn$. \textsf{Algorithm~3} is used to compute $U_n(q)$ and then its $r$-th power is then deduced via binary powering. Finally, the product $M(q^{k})\cdots M(1)$ is again computed using \textsf{Algorithm~3}. The total cost therefore consists of just $O(\M(\sqrt{n}) + \log(N))$ arithmetic operations. It follows that if, for instance, the base field~$\K$ is the prime field~$\F_p$, then the prime number~$p$ should be larger than $N$ in order to exhibit the full strength of the presented algorithms.

\begin{corollary} \label{coro:qexp}
Let $\K$ be a field, 
$q\in \K$ {not a root of unity}, and $N\in \N$.
Let $e_q(x)$ be the $q$-exponential series 
\[e_q(x)\coloneqq\sum_{n \geq 0} \frac{x^n}{[n]_q!},\]
and let $E_q^{{(N)}}(x) \coloneqq e_q(x) \bmod x^N$ be its polynomial truncation of degree~$N-1$.
If $\alpha\in\K$, then one can compute $E_q^{{(N)}}(\alpha)$ in $O(\M(\sqrt{N}))$ operations in~$\K$.
\end{corollary}

\begin{proof}
Denote the summand $\frac{\alpha^n}{[n]_q!}$ by $u_n(q)$. Then  $(u_n(q))_{n\geq 0}$ is $q$-hypergeometric, and satisfies the recurrence $[n+1]_q u_{n+1}(q) - \alpha u_n=0$, therefore $v_N(q)\coloneqq\sum_{i=0}^{N-1} u_{{i}}(q)$ satisfies the second-order recurrence $[n+1]_q (v_{n+2}(q) - v_{n+1}(q)) - \alpha  (v_{n+1}(q) - v_{n}(q))=0$.
Applying Theorem~\ref{thm:qNth} to $v_{N}(q)$ concludes the proof.
\end{proof}

The same result holds true if $e_q(x)$ is replaced by any power series satisfying a $q$-difference equation. For instance, one can evaluate fast all truncations of Heine's  $q$-hypergeometric series 
\[ _2 {\phi}_1([a,b],[c];q;x) \coloneqq \sum_{n \geq 0} \frac{(a;q)_n (b;q)_n}{(c;q)_n} \cdot \frac{x^n}{(q)_n}.\]

\subsection{Complexity analysis and computation of several terms}\label{sec:qNth_complexity}

Theorem \ref{thm:qNth} established an $O(\M(\sqrt{N}))$ cost of the presented method for computing the $N$-th term of a $q$-holonomic sequence. We now aim at performing a detailed complexity analysis with respect to all input parameters. So we need to discuss the complexity of \textsf{Algorithm 3}, where we assume that $M(x) \in \mathcal{M}_n(\K[x]_{d})$ is an $n\times n$ polynomial matrix of degree~$d\geq 1$. We wish to examine the amount of field operations in~$\K$ needed for the computation of $M(q^{N-1}) \cdots M(1)$ in terms of $N,d$ and $n$. Recall that $\M \M(n,d)$ controls the arithmetic complexity of the product in $\mathcal{M}_n(\K[x]_{d})$ and it holds that $\M \M(n,d) = O(n^\theta d+n^2 \, \M(d)) = \tilde{O}(n^\theta d)$.

First, we will examine the direct application of \textsf{Algorithm 3}, where $s = \sqrt{N}$, to $M(x)$. It turns out that the dominating part is step (1), where we compute $P_s(x) = M(q^{s-1}x) \cdots M(x)$ using the divide-and-conquer scheme $P_{2t}(x) = P_t(q^t x) \cdot P_t(x)$ and $P_{2t+1}(x) = M(q^{2t}x) \cdot P_t(q^t x) \cdot P_t(x)$. Note that $P_t(x)$ is an $n\times n$ polynomial matrix of degree at most $td$ and therefore the cost of this step is $O(\M \M(n,sd)) = \tilde{O}(n^\theta d\sqrt{N})$. Step (2) is done component-wisely at each entry of $P(x)$. By the explained fast chirp transform algorithms, it essentially boils down to $n^2$ multiplications of two polynomials, one of degree $sd$ and the other of degree $2sd$. The cost of the second step is therefore $O(n^2 \, \M(s d)) = \tilde{O}(n^2 d \sqrt{N})$. The last step is the multiplication of $N/s=s$ matrices with entries in $\K$ and has therefore an arithmetic complexity of $O(n^\theta s)
=\tilde{O}(n^\theta \sqrt{N})$.

If $d<N$ is a parameter of interest, then there is a better choice of $s$ rather than~$\sqrt{N}$. We saw that the polynomial $P_s(x)$ has degree $sd$ and we must evaluate it at $N/s$ points. The optimal pick for $s$ is therefore $s=\sqrt{N/d}$, which we again can assume to be integer\footnote{Similarly as before, if $\sqrt{N/d}$ is not an integer, then we can compute $u_{N_1}(q)$ first, where $N_1 = \floor{\sqrt{N/d}}^2d$, and then proceed ``naively''. Note that $N-N_1< 2\sqrt{Nd}-d = O(\sqrt{Nd})$.}. Then, by the same arguments as above, the costs of the three steps are $O(\M \M (n,\sqrt{Nd})) = \tilde{O}(n^\theta \sqrt{Nd})$, $O(n^2 \M(\sqrt{Nd}))=\tilde{O}(n^2 \, \sqrt{Nd})$ and $O(n^\theta \sqrt{Nd})$ respectively.

Now, we address specifically the computation of the $N$-th term in a $q$-holonomic sequence. If $(u_n(q))_{n \geq 0}$ is given by a $q$-recurrence 
\[
c_r(q, q^n) u_{n + r}(q) + \cdots  + c_0(q, q^n) u_n(q) = 0,
\]
for $q\in \K$ and for polynomials $c_j(x,y) \in \K[x,y]$, then as observed before, we can compute $u_N(q)$ via 
\[ 
\frac{1}{c_r(q,q^{N-1}) \cdots c_r(q,q)c_r(q,1)} \cdot 
\begin{bmatrix}
0 & \cdots & 0 & 1
\end{bmatrix}
\times \tilde{M}(q^{N-1}) \cdots \tilde{M}(q) \tilde{M}(1) \times \begin{bmatrix}
 1 \\ \vdots \\ 0
\end{bmatrix},
\]
where now
\[
\tilde{M}(x) \coloneqq c_r(q,x) \cdot  M(x) = \begin{bmatrix}
 -c_{r-1}(q,x) & \cdots &  -c_1(q,x) &  -c_0(q,x)
 \\ c_r(q,x) & \cdots & 0 & 0
\\   & \ddots &  & \vdots & 
\\ 0 & \cdots & c_r(q,x) & 0 
\end{bmatrix}. 
\]
Hence, we are interested in $\tilde{M}(q^{N-1})\cdots \tilde{M}(1)$ and $c_r(q,q^{N-1})\cdots c_r(q,1)$. If the degrees of $c_0(q,y),\dots,c_r(q,y)$ are bounded by $d$, then the considerations above imply that the two $q$-factorials can be computed in $O(\M \M(r,\sqrt{Nd}) + r^2 \M(\sqrt{Nd}))$ and $O(\M(\sqrt{Nd}))$  operations in $\K$, respectively. We obtain the following theorem (compare with  \cite[Thm.~2]{BoClSa05}). 
\begin{theorem}\label{thm:detailed_complexity}
Under the assumptions of Theorem~\ref{thm:qNth}, let $d\geq 1$ be the maximum of the degrees of $c_0(q,y),\dots,c_r(q,y)$. Then, for any $N>d$, the term $u_N(q)$ can be computed in $O(r^\theta \sqrt{Nd}+ r^2 \, \M(\sqrt{Nd}))$ operations in $\K$. \\
\end{theorem}

Theorem~\ref{thm:qNth} can be adapted to the computation of \emph{several coefficients} of a $q$-holonomic sequence. The proof is similar to that of Theorem 15 in~\cite{BoGaSc07}, however simpler, because we deal with geometric progressions instead of arithmetic ones.

\begin{theorem} \label{thm:qNth-several}
Under the assumptions of Theorem~\ref{thm:detailed_complexity}, let 
$ N_1 < N_2 < \cdots < N_n = N$ be positive integers, where $n \leq \sqrt{N}$. Then, the terms $u_{N_1}(q),\ldots, u_{N_n}(q)$ can be computed {altogether} in $O(\M(\sqrt{N}) \log N)$ operations in~$\K$.
\end{theorem}

\begin{proof}
As before, we assume that $N$ is a perfect square; let $s=\sqrt{N}$. Examining the presented algorithms, we notice that on the way of computing the matrix $q$-factorial $U_N \coloneqq M(q^{N-1})\cdots M(q)M(1)$ we obtain the evaluated polynomials 
\[
P(1), P(Q), \dots, P(Q^{s-1}),
\]
where $Q\coloneqq q^s$ and $P(x) \coloneqq M(q^{s-1}x) \cdots M(qx)M(x)$. The $q$-factorial $U_N$ is then found by step (3) by trivially multiplying $P(Q^{s-1}) \cdots P(Q)P(1)$. Observe that while multiplying together from right to left we actually also automatically compute
\[
    P(Q^{j-1}) \cdots P(Q)P(1) = M(q^{sj-1})\cdots M(q)M(1) = U_{sj},
\]
for every $j=1,\dots,s$. It follows that employing \textsf{Algorithm~3} and by simply taking the top right element of each $U_{sj}$, we find not only $u_N$, but actually $u_{s}, u_{2s}, \dots, u_{s^2} = u_N$. This already indicates that simultaneous computation of $s$ terms is achievable in similar complexity after some ``distillation''. In general, we are interested in the sequence of $u_i(q)$ at indices $i=N_1,\dots,N_n$, hence we need to perform the following refinement step.

Let $d_0 \in \N$ be a positive integer with $d_0\leq n \leq \sqrt{N}$ and assume that for some $k_1^{(0)},\dots,k_n^{(0)}$ with $k_j^{(0)} \leq N_j < k_j^{(0)}+2d_0$ we already know the values $U_{k_1^{(0)}},\dots,U_{k_n^{(0)}}$. Then we can use a similar strategy as in step (1) of \textsf{Algorithm~3} and deduce the polynomial matrix $P_{d_0}(x) = M(q^{{d_0}-1}x) \cdots M(qx)M(x)$. Compute then the values $q^{k_1^{(0)}},\dots,q^{k_n^{(0)}}$ and evaluate $P_{d_0}(x)$ simultaneously at them. For each $j=1,\dots,n$ it holds that 
\[
P_{d_0}(q^{k_j^{(0)}}) \cdot U_{k_j^{(0)}} = U_{k_j^{(0)}+d_0}.
\] 
We perform this multiplication for those indices $j$ for which $k_j^{(0)} + d_0 \leq N_j < k_j^{(0)}+2d_0$. For these $j$ we then set $k_j^{(1)} = k_j^{(0)}+d_0$ and let $k_j^{(1)} = k_j^{(0)}$ for the other indices; moreover $d_1 \coloneqq \lceil d_0/2 \rceil$. We iterate this process at most $\ell \coloneqq \lceil \log(d_0) \rceil$ many times until $d_\ell=1$. Then we can easily find $U_{N_1},\dots,U_{N_n}$, from which we finally deduce $u_{N_1},\dots,u_{N_n}$. 

Each such step has a cost of at most $O(\M(n)) = O(\M(\sqrt{N}))$ base field operations. Moreover, after first employing \textsf{Algorithm~3} and by the consideration above, we compute $U_1,U_{s},\dots,U_{s^2}$ in $O(\M(\sqrt{N}))$ base operations (Theorem~\ref{thm:detailed_complexity}). Hence, we may choose $d_0 = s$ and for each $j=1,\dots,n$ let $k_j^{(0)}$ be the largest element in $\{1,s,\dots,(s-1)s\}$ such that $k_j^{(0)} \leq N_j$. Clearly, all conditions of the above refinement step are satisfied and we need at most $\lceil \log(s) \rceil = O(\log N)$ many such steps. The total complexity is henceforth 
$O(\M(\sqrt{N})) + O(\M(\sqrt{N}) \log N)=O(\M(\sqrt{N}) \log N)$.
\end{proof}

The same idea applies in the following corollary which states that if $n<\sqrt{N}/N^\varepsilon$ for some $\varepsilon>0$, then we can omit the $\log$-factor in $N$: 
\begin{corollary} \label{cor:qNth-several}
Under the assumptions of Theorem~\ref{thm:detailed_complexity}, let 
$ N_1 < N_2 < \cdots < N_n = N$ be positive integers, where 
$n < N^{\frac12 - \varepsilon}$ for some $0<\varepsilon<\frac12$.
Then, the terms $u_{N_1}(q),\ldots, u_{N_n}(q)$ can be computed 
{altogether} in $O(\M(\sqrt{N}))$ operations in~$\K$.
\end{corollary}
\begin{proof}
Here, we follow the exact same procedure as in the proof before. Then, regarding complexity, we use $O(\M(n)) = O(\M(N^{\frac12 - \varepsilon})) = O(\M(\sqrt{N})N^{-\varepsilon})$ and obtain that the same method yields a total arithmetic cost of $O(\M(\sqrt{N})) + O(\M(\sqrt{N})N^{-\varepsilon} \log N) =O(\M(\sqrt{N})).$
\end{proof}

Remark that regarding a detailed complexity analysis for the computation of several coefficients, we have the following trade-off: either we compute at most $\sqrt{N}$ terms in the arithmetic complexity $O((r^\theta d\sqrt{N}+ r^2 \M(d\sqrt{N})) \log N)$, or at most $\sqrt{N/d}$ terms, but in a better cost of $O((r^\theta \sqrt{Nd}+ r^2 \M(\sqrt{Nd})) \log N)$. The proofs combine the considerations above with setting $s=\sqrt{N}$ and $s=\sqrt{N/d}$ respectively. In both cases we can get rid of the log-factor in $N$ like in Corollary~\ref{cor:qNth-several} by computing a factor of $N^\varepsilon$ less terms.

\subsection{The case $q$ is an integer: bit complexity}\label{sec:BinSplit}

Until now, we only considered the arithmetic complexity model, which is very well-suited
to measure the algorithmic cost when working in algebraic structures
whose basic internal operations have constant cost (such as finite fields, or floating point numbers).

Now we discuss here the case where $q$ is an integer (or rational) number. The
arithmetic complexity model needs to be replaced by the bit-complexity model.

Recall that the most basic operations on integer numbers can be performed in quasi-optimal time, that is, in a number of bit operations which is almost linear, up to logarithmic factors, in (the maximum of) their bit size. 
The most basic operation is integer multiplication, for which quasi-linear time algorithms
are known since the early seventies, starting with the famous paper by Sch\"{o}nhage and Strassen~\cite{ScSt71} who showed that two $n$-bit integers can be multiplied in
$O(n \log n \log \log n)$ bit operations.
After several successive improvements, e.g.,~\cite{Furer09,HaHoLe16}, we know as of 2020 that  two $n$-bit integers can be multiplied in time $O(n \log n)$~\cite{HaHo20}. We shall call the cost of multiplying two $n$-bit integers $\M_\Z(n)$.

In this context, the matrix $q$-factorials from~\S\ref{sec:prelim} are computed by \emph{binary splitting} rather than by baby-step/giant-steps.
Recall that this phenomenon already occurs in the usual holonomic setting. For example, the bitsize of $u_N=N!$ is $O(N \log N)$, however both, the ``naive'' method of computing it using $u_n=n u_{n-1}$, or Strassen's baby-steps/giant-steps method, yield worse bit complexity. In the naive approach the problem is that integers of unbalanced bitsize are multiplied together and hence not the full power of fast integer multiplication techniques can be employed. A more clever and very simple way is to just use the fact that 
\[
N! = (1 \cdots \floor{N/2}) \times ((\floor{N/2}+1)\cdots N),
\]
and that the bitsizes of both factors have magnitude $O(N/2 \log(N)) = \tilde{O}(N)$. Thus, fast integer multiplication can be used to multiply them in $\tilde{O}(N)$ bit-complexity. This idea results in the binary splitting \textsf{Algorithm \ref{alg:binsplit}}. This algorithm is very classical, and we only recall it for completeness. 
See~\cite[{\S}12]{Bernstein08} for a good survey on 
this technique, and its applications.
\setcounter{algorithm}{4}
\begin{algorithm}[h]
\begin{algorithmic}[1]
\Function{$\mathcal{F}$}{$A$}
\If{$N=1$}
\State \Return $A[1]$
\EndIf
\State \Return $\mathcal{F}(A[\floor{N/2}+1,\dots,N]) \cdot \mathcal{F}(A[1,\dots,\floor{N/2}]) $
\EndFunction
\end{algorithmic}
\caption{(\textcolor{red}{\textsf{BinSplit}}) \\
\quad {\bf Input}: $A=[a_1,\dots,a_N]$ \hfill list of elements from some arbitrary ring $R$
\\{\bf Output}: $a_N \cdots a_1$}
\label{alg:binsplit}
\end{algorithm}  

Assume that each $a_i$ in the input of \textsf{BinSplit} has at most $k$ bits and let $C(n)$ be the complexity of \textsf{BinSplit} if $A = [a_1,\dots,a_n]$ is $n$-dimensional. It follows that
\[
C(N) \leq 2C(\lceil N/2 \rceil) + \M_\Z(N/2 \cdot k),
\]
where $\M_\Z(n)= \tilde{O}(n)$ is the cost of multiplication of integers with $n$ bits. We obtain $C(N) = \tilde{O}(Nk)$. Hence, using this method, the computation of $u_N = N!$ has $\tilde{O}(N)$ bit-complexity, which is quasi-optimal. Moreover, the same idea applies to any holonomic sequence, by deducing the first order matrix recurrence and computing the matrix product using \textsf{BinSplit}. 

Now we shall see that the $q$-holonomic case is similar. First, let $q$ be a positive integer of $B$ bits and consider the computation of the $q$-factorial 
\[
u_N(q) = (1+q)(1+q+q^2)\cdots(1+q+\cdots+q^{N-1}),
\]
as an illustrative example. For each factor we have the trivial inequalities 
\[
q^n < 1+q+\cdots+q^n < q^{n+1},
\]
meaning that $q^{N(N-1)/2}< u_N(q)<q^{N(N+1)/2}$, so the bitsize of $u_N(q)$ is of magnitude~$N^2B$. The ``naive'' algorithm of deducing $u_N(q)$ by first computing the integers $q^i$, then the corresponding sums and products, has $\tilde{O}(N^3B)$ binary complexity. This method is not (quasi-)optimal with respect to the output size. It is also easy to see that the presented baby-steps/giant-steps based algorithms yield bad bit-complexity as well, despite their good arithmetic cost.

Similarly, if
\[
u_N(q) = \sum_{n=0}^{N-1} q^{n^2},
\]
then the integer $u_N(q)$ is bounded in absolute value from above by $N q^{(N-1)^2}$ and by $q^{(N-1)^2}$ from below, so its bitsize is again of magnitude~$N^2 B$. The ``naive'' algorithm consisting of computing the terms $q^{i}$ one after the other before summing, has again non-optimal bit-complexity $\tilde{O}(N^3 B)$.

Can one do better? The answer is ``yes'' and one can even achieve a complexity which is quasi-linear in the bitsize of the output. Similarly to the holonomic setting, it is sufficient to use the $q$-holonomic character of $u_N(q)$, and to reduce its computation to that of a $q$-factorial matrix as in~\S\ref{sec:NoSc}, which can then be handled with \textsf{BinSplit}. To  be more precise, given an integer or rational number $q$ of bitsize $B$ and any $q$-holonomic sequence $(u_n(q))_{n\geq0}$ defined by polynomials $c_0,\dots,c_r \in \Z[x,y]$ with $c_r(q,q^n) \neq 0$ for any $n \in \N$, we define $M(x) \in \Q(x)$ by 
\[
\begin{bmatrix}
 -\frac{c_{r-1}(q,x)}{c_r(q,x)} & \cdots &  -\frac{c_1(q,x)}{c_r(q,x)} &  -\frac{c_0(q,x)}{c_r(q,x)}
\\ 1 & \cdots & 0 & 0
\\ \vdots  & \ddots & \vdots & \vdots & 
\\ 0 & \cdots & 1 & 0 
\end{bmatrix}.
\]
Then, as observed before, $u_N$ can be read off from
\[
{M}(q^{N-1})\cdots {M}(q){M}(1),
\]
which we aim to compute efficiently. Again, instead of using baby-steps/giant-steps, it is a better idea to use binary splitting by applying \textsf{Algorithm \ref{alg:binsplit}} to $A = [M(1),\dots,M(q^{N-1})]$. Note that obviously, any element in $A$ is a matrix with rational entries of bitsize bounded by $O(NB)$. Therefore, the complexity of \textsf{BinSplit} does not exceed $\tilde{O}(N^2B)$ by the same argument as before, now using fast multiplication of rational numbers. These considerations prove 
\begin{theorem} \label{thm:qNth-bin}
Under the assumptions of Theorem~\ref{thm:qNth}, with $\K=\Q$,
the term $u_N(q)$ can be computed in $\tilde{O}(N^2 B)$ bit operations,
where $B$ is the bitsize of~$q$.
\end{theorem}

As a corollary, (truncated) solutions of $q$-difference equations can be
evaluated using the same (quasi-linear) bit-complexity. This result should be
viewed as the $q$-analogue of the classical fact that holonomic functions can
be evaluated fast using binary splitting, a 1988 result by the Chudnovsky
brothers~\cite[{\S}6]{ChuChu88}, anticipated a decade earlier (without proof)
by Schroeppel and Salamin in Item~178 of~\cite{hakmem}.

%%%%%%%%%%%%%%%%%%%%%%%%%%%%%%%%%%%%%%%%%%%%%%%%%%%%%%%%%%%%%%%%%%%%
\section{Applications}\label{sec:app}
%%%%%%%%%%%%%%%%%%%%%%%%%%%%%%%%%%%%%%%%%%%%%%%%%%%%%%%%%%%%%%%%%%%%

\subsection{Combinatorial $q$-holonomic sequences}

As already mentioned, many $q$-holonomic sequences arise in combinatorics, for
example in connection with the enumeration of lattice polygons, where
$q$-analogues of the Catalan numbers $\frac{1}{n+1} \binom{2n}{n}$ occur
naturally~\cite{Gessel80,FuHo85},
or in the enumeration of
special families of matrices with coefficients
in {the finite field}~$\F_q$~\cite{Kirillov95,KiMe97,Yip18},
where sequences related to the
Gaussian coefficients $\binom{n}{k}_q$ also show up.

A huge subfield of combinatorics is the theory of
partitions~\cite{Andrews1976}, where $q$-holonomic sequences occur as early as
in the famous Rogers-Ramanujan identities~\cite{Rogers1893,RoRa19}, see also~\cite[Ch.~7]{Andrews1976}, 
{\emph{e.g.},}
\[1 + \sum_{n \geq 1} \frac{q^{n^2}}{(1-q) \cdots (1-q^n)}
=
\prod_{n \geq 0} \frac{1}{(1-q^{5n+1})(1-q^{5n+4})}
\]
which translates the fact that the number of partitions of $n$ into parts that
differ by at least $2$ is equal to the number of partitions of~$n$ into parts
congruent to $1$ or $4$ modulo $5$. Andrews~\cite{Andrews72,Andrews74}, see also~\cite[Chapter~8]{Andrews1976}, 
laid the foundations of a
theory able to capture the $q$-holonomy of any generating function of a
so-called \emph{linked partition ideal}.

As a consequence, a virtually infinite number of special families of
polynomials coming from partitions can be evaluated fast. For instance, 
the family of truncated polynomials
\[ F_n(x) \coloneqq \prod_{k=1}^\infty (1-x^k)^{{3}} \bmod x^n,\]
can be evaluated fast due to our results and to the identity~\cite[\S6]{Pak06}
\[
F_N(q) 
= \sum_{\binom{n+1}{2}<N} (-1)^n (2n+1) q^{\binom{n+1}{2}}
.
\]

\subsection{Evaluation of $q$-orthogonal polynomials} In the theory of special
functions, \emph{orthogonal polynomials} play {a fundamental} role. There exists
an extension to the $q$-framework of the theory, see \emph{e.g.}, Chapter 9 in
Ernst's book~\cite{Ernst12}. Amongst the most basic examples, the
\emph{discrete $q$-Hermite polynomials}~\cite{AlCa65,Askey89} 
are defined by
their $q$-exponential generating function
\[\sum_{n \geq 0} F_{n,q}(x)\frac{t^n}{[n]_q!} = \frac{e_q(xt)}{e_q(t)e_q(-t)},
\]
and therefore they satisfy the second-order linear $q$-recurrence
\[
F_{n+1,q}(x) = x F_{n,q}(x) - (1-q^n) q^{n-1} F_{n-1,q}(x), \quad  n\geq 1, 
\]
with initial conditions $F_{0,q}(x)=1, F_{1,q}(x)=x$. From there, it follows
that for any $\alpha\in\K$, the sequence $(F_{n,q}(\alpha))_{n\geq 0}$ is $q$-holonomic,
thus the evaluation of the $N$-th polynomial at $x=\alpha$ can be
computed fast. The same is true for the \emph{continuous $q$-Hermite
polynomials}, for which $2 \alpha H_{n,q}(\alpha) = H_{n+1,q}(\alpha) + (1-q^n) H_{n-1,q}(\alpha)$
for $n\geq 1$, and $ H_{0,q}(\alpha)=1, H_{1,q}(\alpha)=2\alpha$. More generally, our results
in \S\ref{sec:main} imply that any family of $q$-orthogonal polynomials can be
evaluated fast.

%%%%%%%%%%%%%%%%%%%%%%%%%%%%%%%%%%%%%%%%%%%%%%%%%%%%%%%%%%%%%%%%%%%%
\subsection{Polynomial and rational solutions of $q$-difference equations} \label{sec:rat_sols_q-diff_eq}
%%%%%%%%%%%%%%%%%%%%%%%%%%%%%%%%%%%%%%%%%%%%%%%%%%%%%%%%%%%%%%%%%%%%

The computation of polynomial and rational solutions of linear differential
equations lies at the heart of several important algorithms, for computing
hypergeometric, d'Alembertian and Liouvillian solutions, for factoring and for computing
differential Galois groups~\cite{PuSi03,AbZi96,APP98}. 
Creative telescoping algorithms (of
second generation) for multiple integration with
parameters~\cite{Chyzak00,Koutschan10} also rely on computing rational
solutions, or deciding their existence. The situation is completely similar
for $q$-difference equations, \emph{i.e.} equations of the form 
\begin{align}\label{eq:diff_eq}
    Ly = a_\nu(x) y(q^\nu x) + a_{\nu-1}(x) y(q^{\nu-1} x) + \cdots +  a_0(x)y(x) = 0,
\end{align}
with $a_j(x) \in \Q[x]$, for all $j=0,\dots,\nu$, such that $a_0(x)a_\nu(x)$ is not identically zero. Improving algorithms for polynomial and rational solutions of such equations is important for finding {$q$}-hypergeometric solutions~\cite{APP98}, for computing $q$-difference
Galois groups~\cite{Hendriks97,ArZh20}, and for performing $q$-creative telescoping~\cite{Koornwinder93,Chyzak00,Koutschan10}.

In both differential and $q$-differential cases, algorithms for computing
polynomial solutions proceed in two distinct phases: (i)~compute a degree
bound~$N$, potentially exponentially large in the equation size; (ii)~reduce
the problem of computing polynomial solutions of degree at most~$N$ to linear
algebra. Abramov, Bronstein and Petkov{\v s}ek showed in~\cite{AbBrPe95} that,
in step (ii), linear algebra in size~$N$ can be replaced by solving a much
smaller system, of polynomial size. However, setting up this smaller system
still requires linear time in~$N$, essentially by unrolling a ($q$-)linear
recurrence up to terms of indices close to~$N$. For differential (and
difference) equations, this step has been improved in~\cite{BoClSa05,BCCS06},
by using Chudnovskys' algorithms for computing fast the $N$-th term of a
holonomic sequence. This allows for instance to decide (non-)existence of
polynomial solutions in sublinear time~$\tilde{O}(\sqrt{N})$. Moreover, when
polynomial solutions exist, one can represent/manipulate them in
\emph{compact~form} using the recurrence and initial terms as a compact
data structure. Similar ideas allow to also compute rational solutions in compact form in the same complexity, see~\cite[Chap.~17]{AECF}.

The same improvements can be transferred to linear $q$-difference equations, in
order to improve the existing algorithms~\cite{Abramov95,AbBrPe95,Khmelnov00}.
In this case,
setting up the smaller system in phase~(ii) amounts to computing the
$N$-th term of a {$q$-holonomic} sequence, and this can be done 
fast using our
results in \S\ref{sec:main}.
{A technical subtlety is that, as pointed out in~\cite[\S4.3]{AbBrPe95}, it is
not obvious in the $q$-difference case how to guarantee the non-singularity
of the $q$-recurrence on the coefficients of the solution. This induces potential technical
complications similar to the ones for polynomial solutions of differential
equations in small characteristic, which can nevertheless be overcome by adapting
the approach described in~\cite[\S3.2]{BoSc09}.}
Similar improvements can be also transferred to systems~\cite{Abramov02,BaClHa19}.

Let us finish this discussion by pointing out briefly an application of these improvements.
Desingularizing a linear differential operator $L(x,\partial_x)$ consists in computing a left multiple with all apparent singularities removed. It is a central task for determining the Weyl closure of~$L$~\cite{Tsai00}. The computation of polynomial solutions of Fourier dual operators is a basic step for performing desingularizations~\cite{ChDuMaMiSa16}. By duality, the order~$N$ of the desingularization corresponds to the degree of polynomial solutions of the dual~$L^\star$ of~$L$.
This remark in conjunction with the fast algorithms for polynomial solutions~\cite{BoClSa05}, themselves based on the fast computation of matrix factorials, allows to speed up the computation of desingularizations.
The situation is similar in other Ore algebras, and in particular for $q$-difference equations. Therefore, our algorithmic improvements in the computation of polynomial solutions of 
 $q$-difference equations,  themselves based on the fast computation of matrix $q$-factorials, have a direct impact on the acceleration of the desingularizations process for $q$-difference equations. It is less obvious to us whether other desingularization algorithms, such as the one from~\cite{KoZh18}, could also benefit from these remarks.

%%%%%%%%%%%%%%%%%%%%%%%%%%%%%%%%%%%%%%%%%%%%%%%%%%%%%%%%%%%%%%%%%%%%
\subsection{Computing curvatures of $q$-difference equations}\label{sec:q-diff_curv}
%%%%%%%%%%%%%%%%%%%%%%%%%%%%%%%%%%%%%%%%%%%%%%%%%%%%%%%%%%%%%%%%%%%%
A natural application of the fast computation of matrix $q$-factorials is the computation of curvatures of $q$-difference equations, since in this area these objects appear quite inherently. Another strong motivation comes from the fact that the $q$-analogue~\cite{Bezivin91} of Grothendieck's conjecture (relating equations over~$\Q$ with their reductions modulo primes~$p$) is proved~\cite{Vizio02, DiHa20}, while the classical differential case is widely open~\cite{Katz72}. Still, in the latter setting algorithms have been developed allowing to compute $p$-curvatures fast~\cite{BoSc09,BoCaSc14,BoCaSc15,BoCaSc16}. They allow, for example, to perform a quick heuristic, but reliable in practice, test for the existence of a basis of algebraic solutions of a linear differential operator~\cite{BoKa09}. 

The $q$-analogue of Grothendieck's conjecture investigates rational solutions of $q$-difference equations. Similarly to the differential setting, one naturally associates to the equation (\ref{eq:diff_eq}) the linear $q$-difference system $Y(qx) = A(x) Y(x)$, where 
\[ 
A(x) = \begin{bmatrix}
        0       & 1 & \cdots & 0\\
        \vdots  &   & \ddots &  \\
        0       & 0 & \cdots & 1\\
        -a_0/a_\nu & -a_1/a_\nu &  \cdots & -a_{\nu-1}/a_\nu
    \end{bmatrix}.
\]
Then it easily follows that if $z(x)$ solves (\ref{eq:diff_eq}), then $(z(x),z(qx),\dots,,z(q^{\nu-1}x))^t$ is a solution of $Y(qx) = A(x) Y(x)$. Moreover, these equations are in some sense equivalent through the $q$-analogues of the Wronskian Lemma and the Cyclic Vector Lemma whenever $q$ is not a root of unity or order smaller than $\nu$; see \cite{DiHa20} for details. If $q$ is considered as a variable, then Di Vizio and Hardouin proved that (\ref{eq:diff_eq}) admits a full set of solutions in $\Q(q,x)$ if and only if for almost all natural numbers $n$,
\[
C_n(x) \coloneqq A(q^{n-1}x)\cdots A(qx) A(x) \equiv \Id_\nu \mod \operatorname{GL}_\nu(R_n(x)),
\]
where $R_n = \Q[q]/\Phi_n(q)$, with $\Phi_n(x)$ the $n$-th cyclotomic polynomial. The elements in the sequence $(C_n(x))_{n\geq1}$ are known as curvatures of the $q$-difference system and are clearly just matrix $q$-factorials.

On the other hand, if $q \in \Q$ then it already follows from main result of \cite{Vizio02} that (\ref{eq:diff_eq}) admits a basis of rational solutions in $\Q(x)$ if and only if for almost all primes $p$,
\[
A(q^{\kappa_p-1} x)\cdots A(qx)A(x) \equiv \Id_\nu \mod p^\ell, 
\]
where $\kappa_p = \text{ord}_p(q)$ and $\ell_p \in \mathbb{Z}$ such that $1-q^{\kappa_p} = p^{\ell_p}\frac{h}{g}$, with $h,g \in \mathbb{Z}$ coprime to $p$. 

On these types of questions there is more progress in the $q$-difference setting than in the classical differential one. Yet, unfortunately, the theorems above are not proven to be effective in the sense that still infinitely many conditions need to be checked in order to conclude the implication we are mostly interested in. Therefore, the computation of any finite number of curvatures only provides a heuristic for the existence of rational solutions of a $q$-difference equation. Moreover, the mentioned algorithms in \S\ref{sec:rat_sols_q-diff_eq} compute rational solutions of equations of type (\ref{eq:diff_eq}) and therefore allow to decide rigorously about the existence of such a basis. However, all these methods have a cost which is potentially exponential in the size of the input. Our goal in this section is to design a fast heuristic test for the existence of a basis of rational solutions of a $q$-difference equation using curvatures.

If $q$ is a variable, we want to check whether $C_n(x) \equiv \Id_\nu$ mod $\operatorname{GL}_\nu(R_n(x))$ for many~$n$. Clearly, after the reduction mod $\Phi_n(q)$, the polynomial $C_n(x)$ has arithmetic size $n^2$ over~$\Q$ and $\tilde{O}(n^3)$ bitsize. Hence, computing $C_n(x)$ is unnecessarily costly. We propose to work over $R_{n,p} \coloneqq \F_p[q]/\Phi_n(q)$ for some (large) prime $p$; moreover, we compute $C_n(x_0)$ for some randomly chosen $x_0 \in \F_p$. Furthermore, in order to avoid computing general cyclotomic polynomials, we compute $C_n(x)$ only for prime numbers $n$. After these considerations, it is easy to see that \textsf{Algorithm~3} applies and allows to deduce $C_n(x_0)$ modulo $\operatorname{GL}_\nu(R_{n,p})$ in $\tilde{O}(\sqrt{n})$ arithmetic operations in $R_{n,p}$, hence $\tilde{O}(n^{3/2})$ operations in $\F_p$. Finally, if we want to deduce this quantity for all primes $n$ between $2$ and some $N \in \N$, it is wiser to apply the accumulating remainder tree method presented in \cite{CoGeHa14, Harvey14}, which allows for quasi-optimal complexity of $\tilde{O}(N^{2})$ in this case.

If $q$ is some rational number and the goal is to test whether $Y(qx) = A(x)Y(x)$ has a full set of rational solutions, one may check 
\[
C_{\kappa_p}(x) = A(q^{\kappa_p-1} x)\cdots A(qx)A(x) \equiv \Id_\nu \mod p^\ell, 
\]
for many primes $p$. Unfortunately, finding the order $\kappa_p$ in practice may be costly, therefore we shall check the weaker assumption $C_{p-1}(x) \equiv \Id_\nu$ mod $p$. Conjecturally, this equality for sufficiently many primes $p$ is also enough to conclude on the existence  of a basis of rational solutions. The presented \textsf{Algorithm~3} allows to compute $C_{p-1}(x_0)$ mod $p$ for some $x_0 \in \F_p$ in $\tilde{O}(\sqrt{p})$ arithmetic cost. Finally, again, if we choose $N,x_0 \in \N$, then the accumulating remainder tree allows to deduce $C_{p-1}(x_0)$ mod $p$ for all primes $p$ between $2$ and $N$ quasi-optimally in $\tilde{O}(N)$ bit operations.

%%%%%%%%%%%%%%%%%%%%%%%%%%%%%%%%%%%%%%%%%%%%%%%%%%%%%%%%%%%%%%%%%%%%
\subsection{$q$-hypergeometric creative telescoping}
%%%%%%%%%%%%%%%%%%%%%%%%%%%%%%%%%%%%%%%%%%%%%%%%%%%%%%%%%%%%%%%%%%%%

In the case of differential and difference hypergeometric creative
telescoping, it was demonstrated in~\cite{BCCS06} that the 
compact representation for polynomial solutions can be used as an efficient
data structure, and can be applied to {speed up} the computation of Gosper
forms and Zeilberger's classical summation algorithm~\cite[\S6]{PWZ96}. The
key to these improvements lies in the fast computation of the $N$-th term of a
holonomic sequence, together with the close relation between Gosper’s
algorithm and the algorithms for rational solutions.

Similarly, in the $q$-difference case, Koornwinder's $q$-Gosper
algorithm~\cite[\S5]{Koornwinder93} is closely connected to Abramov's
algorithm for computing rational solutions~\cite[\S2]{Abramov95}, and this
makes it possible to transfer the improvements for rational solutions to the
$q$-Gosper algorithm. This leads in turn {to} improvements upon
Koornwinder's algorithm for $q$-hypergeometric summation~\cite{Koornwinder93},
along the same lines as in the differential and difference {cases}~\cite{BCCS06}.

%%%%%%%%%%%%%%%%%%%%%%%%%%%%%%%%%%%%%%%%%%%%%%%%%%%%%%%%%%%%%%%%%%%%
\section{Experiments}\label{sec:exp}		%%%%%%%%%%%%%%%%%%%%%%%%%%%%%%%%%%%%%%%%%%%%%%%%%%%%%%%%%%%%%%%%%%%%

{\sf Algorithms 1} and~{\sf 2} were implemented in \textcolor{magenta}{\href{http://magma.maths.usyd.edu.au}{\sf Magma}} and \textsf{Algorithm 3} in  \textcolor{magenta}{\href{https://maplesoft.com}{\sf Maple}}. All implementations deliver some encouraging timings. Of course, since these algorithms are designed to be fast in the \emph{arithmetic model}, it is
natural to make experiments over a finite field $\mathbb{K}$, or over
truncations of real/complex numbers, as was done in~\cite{NognengSchost18} for
the problem in \S\ref{sec:NoSc}. 

Recall that both {\sf Algorithms 1} and~{\sf 2} compute $\prod_{i=0}^{N-1} (\alpha - q^i) \; \in \mathbb{K}$, given $\alpha,q$
in a field~$\mathbb{K}$, and $N\in\mathbb{N}$, whereas \textsf{Algorithm~3} finds $M(q^{N-1})\cdots M(q) M(1) \in \K^{n \times n}$ for a given polynomial matrix $M(x) \in \K[x]$ of size $n \times n$ and $q \in \K, N\in\N$. In our experiments, $\mathbb{K}$ is the finite field $\mathbb{F}_p$ with $p=2^{30}+3$ elements.

Timings for {\sf Algorithms 1} and~{\sf 2} are presented in Table~\ref{tab:timings}. We compare the straightforward iterative algorithm (column~{\sf Naive}), to the fast
baby-step/giant-step algorithms, one based on subproduct trees and
resultants (column~{\sf Algorithm 1}), the other based on multipoint
evaluation on geometric sequences (column~{\sf Algorithm~2}).

\begin{table}[t]
  \centering
\begin{tabular}{|c|S|S|S|}
  \hline
 degree $N$ & {\sf Naive algorithm} & {\sf Algorithm 1} & {\sf Algorithm 2}\\
  \hline
  $2^{16}$ & 0.04 & 0.03 & 0.00 \\
  $2^{18}$ & 0.18 & 0.03 & 0.01\\
  $2^{20}$ & 0.72 & 0.06 & 0.01\\
  $2^{22}$ & 2.97 & 0.14 & 0.02\\
  $2^{24}$ & 11.79 & 0.32 & 0.04\\
  $2^{26}$ & 47.16 & 0.73 & 0.08 \\
  $2^{28}$ & 188.56 & 1.68 & 0.15 \\
  $2^{30}$ & 755.65 & 3.84 & 0.31 \\
  $2^{32}$ & 3028.25 
		& 8.65 & 0.64 \\
\hline
  $2^{34}$ & 
		& 19.65 & 1.41 \\
  $2^{36}$ & 
		& 44.42 & 2.96 \\
  $2^{38}$ &  
		& 101.27 & 6.36 \\
  $2^{40}$ & 
		& 228.58  & 14.99 \\
  $2^{42}$ &  & 515.03 & 29.76 \\
  $2^{44}$ &  & 1168.51 & 61.69 
				\\
  $2^{46}$ &  & 2550.28  & 137.30 
				\\
\hline
  $2^{48}$ &  & 
		& 297.60  
				\\
  $2^{50}$ &  & 
		& 731.63  
				\\
  $2^{52}$ &  &  & 1395.33 
				\\
  $2^{54}$ &  &  & 3355.39 \\
  \hline
\end{tabular}
\caption{
 Comparative timings (in seconds) for the computation of $\prod_{i=0}^{N-1} (\alpha - q^i) \; \in \mathbb{F}_{p}$, with 
$p=2^{30}+3 
$ 
and $(\alpha,q)$ randomly chosen in $\mathbb{F}_{p}\times \mathbb{F}_{p}$.
All algorithms were executed on the same machine, running {\sf Magma v. 2.24}. 
For each target degree~$N$, each execution was limited to 1~hour.
{\sf Naive algorithm} could reach degree $N=2^{32}$, {\sf Algorithm~1} degree $N=2^{46}$, and {\sf Algorithm~2} degree $N=2^{54} = 
8\,014\,398\,509\,481\,984$.
By extrapolation,
the {\sf Naive algorithm} would have needed
$\approx 4^{11} \times 3028.25$ sec. 
$\approx 400$~years on the same instance, and {\sf Algorithm 2} approximately $18$ hours. 
}
\label{tab:timings}
\end{table}

Some conclusions can be drawn by analyzing these timings:
\begin{itemize}
\vspace{-0.1cm}
  \item The theoretical complexities are perfectly reflected in practice: as $N$ is increased from $2^{2k}$ to $2^{2k+2}$, timings are also multiplied (roughly) by 4 in column {\sf Naive}, and (roughly) by $2$ in columns
{\sf Algorithm 1} and {\sf Algorithm 2}.
  \item The asymptotic regime is reached from the very beginning.
  \item {\sf Algorithm 2} is always faster than {\sf Algorithm 1}, which is itself much faster than the {\sf Naive algorithm}, as expected.
  \item A closer look into the timings shows that for {\sf Algorithm 1}, $\approx 80\%$ of the time is spent in step~(3) (resultant computation), the other steps taking $\approx 10\%$ each; for {\sf Algorithm 2}, step (1) takes $\approx 25\%$, step~(2)
takes $\approx 75\%$, and step~(3) is negligible.
\end{itemize}

In order to visualize the performance of \textsf{Algorithm~3}, we took random polynomial matrices in $\mathcal{M}_n(\F_p[x])$ of degree $d$. Figure~\ref{fig:n} compares the time needed to compute the Matrix $q$-factorial for $d=1$ and $n=2,4,8$ as $N$ growths. The black lines represent the best linear fits to the data points and are given by the equations $y=0.56x-14.1, y=0.59x-12.5$ and $y=0.53x-9.0$ respectively. Note that we are plotting on a log-log scale, therefore the established complexity of $\tilde{O}(N^{1/2})$ is indicated by the coefficient of $x$ in these linear fits which is always only slightly greater than $1/2$. In the same figure, we also show the timings for $d=1$ and $n=2$ of the naive algorithm given by successively computing and multiplying $M(q^i)$ together. The best linear fit almost perfectly describes this data and has a slope of $1.04 \approx 1$, in line with the linear complexity in $N$. 
\begin{figure}[t] 
\centering
\begin{tikzpicture}[scale=1]
\begin{axis}[
    xmode=log,
    ymode=log, 
    log basis x={2},
    log basis y={2},
    xlabel={$N$},
    ylabel={Time in seconds},
    xmin=2^9, xmax=2^39,
    ymin=0, ymax=260,
    xtick={2^10,2^14,2^18,2^22,2^26,2^30,2^34,274000000000},
    ytick={2^(-6),2^(-4),2^(-2),1,2^2,2^4,2^6,2^8},
    legend pos=south east,
    ymajorgrids=true,
    grid style=dashed,
]
\addplot[
    color=blue,
    mark=x,
    ]
    coordinates {(2^12,.125) (2^14,.485) (2^16,1.875) (2^18,9.640) (2^20,39.516) (2^22,160.750)};
\addplot[
    color=blue,
    mark=o,
    ]
    coordinates { (2^14,.15e-1) (2^16,.31e-1) (2^18,.62e-1) (2^20,.125) (2^22,.234) (2^24,.656) (2^26,1.156) (2^28,2.359) (2^30,5.156) (2^32,12.421) (2^34,27.546) (2^36,69.578) (2^38,189.218)};
\addplot[
    color=green,
    mark=square,
    ]
    coordinates { (2^12,.31e-1) (2^14,.62e-1) (2^16,.109) (2^18,.218) (2^20,.468) (2^22,1.468) (2^24,3.093) (2^26,7.734) (2^28,16.640) (2^30,34.140) (2^32,110.406) (2^34,178.953) };
\addplot[
    color=red,
    mark=triangle,
    ]
    coordinates { (2^10,.62e-1) (2^12,.171) (2^14,.328) (2^16,1.062) (2^18,1.484) (2^20,3.796) (2^22,6.453) (2^24,11.031) (2^26,24.125) (2^28,62.046) (2^30,143.640) };
\addplot[
    color=black,
    ]
    coordinates { (2^12, 0.117409) (2^22, 162.6235)};
\addplot[
    color=black,
    ]
    coordinates { (2^12, 0.005828) (2^38, 136.74)};
\addplot[
    color=black,
    ]
    coordinates { (2^12, 0.02373) (2^34, 190.0717)};
\addplot[
    color=black,
    ]
    coordinates { (2^10, 0.07837233 ) (2^30, 127.7553)};

\legend{$d=1;n=2$; naive}
\addlegendentry{$d=1$; $n=2$}
\addlegendentry{$d=1$; $n=4$}
\addlegendentry{$d=1$; $n=8$}
\end{axis}
\end{tikzpicture}
\caption{Timings of \textsf{Algorithm~3} implemented in \textsf{Maple~2020.2}. We compare $d=1$ and $n=2,4,8$ for various values of $N$.}
\label{fig:n}
\end{figure}
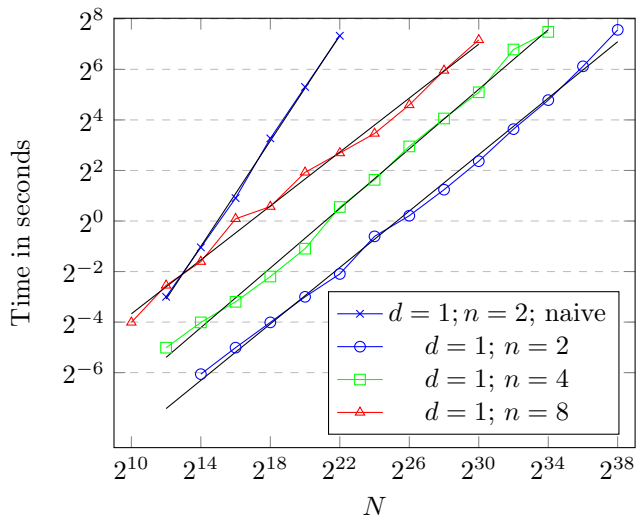

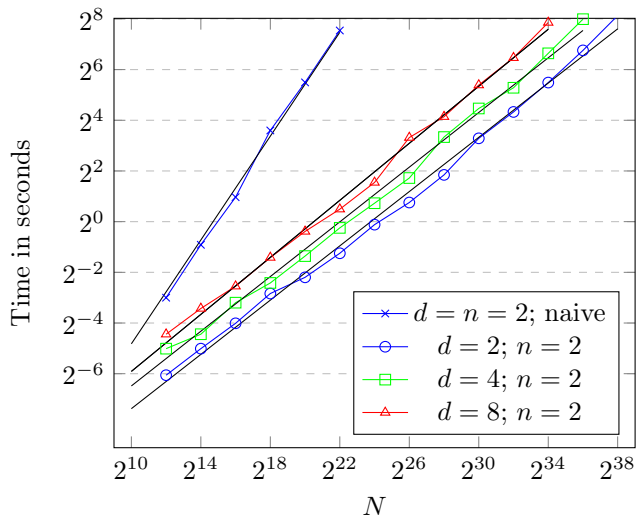
\begin{figure}[h] 
\centering 
\begin{tikzpicture}[scale=1]
\begin{axis}[
    xmode=log,
    ymode=log, 
    log basis x={2},
    log basis y={2},
    xlabel={$N$},
    ylabel={Time in seconds},
    xmin=2^9, xmax=2^39,
    ymin=0, ymax=260,
    xtick={2^10,2^14,2^18,2^22,2^26,2^30,2^34,274000000000},
    ytick={2^(-6),2^(-4),2^(-2),1,2^2,2^4,2^6,2^8},
    legend pos=south east,
    ymajorgrids=true,
    grid style=dashed,
]
\addplot[
    color=blue,
    mark=x,
    ]
    coordinates {(2^12,.125) (2^14,0.531) (2^16,1.953) (2^18,12.110) (2^20,45.062) (2^22,186.438) };
\addplot[
    color=blue,
    mark=o,
    ]
    coordinates {(2^12,.15e-1) (2^14,.31e-1) (2^16,.62e-1) (2^18,.140) (2^20,.218) (2^22,.421) (2^24,0.921) (2^26,1.687) (2^28,3.593) (2^30,9.765) (2^32,20.093) (2^34,44.875) (2^36,108.203) (2^38,280.593)};
\addplot[
    color=green,
    mark=square,
    ]
    coordinates {(2^12,.31e-1) (2^14,.46e-1) (2^16,.109) (2^18,.187) (2^20,.390) (2^22,.843) (2^24,1.656) (2^26,3.296) (2^28,10.109) (2^30,22.093) (2^32,38.937) (2^34,99.828) (2^36,254.093)};
\addplot[
    color=red,
    mark=triangle,
    ]
    coordinates {(2^12,.46e-1) (2^14,.93e-1) (2^16,0.171) (2^18,.375) (2^20,.765) (2^22,1.406) (2^24,2.921) (2^26,9.968) (2^28,17.640) (2^30,41.906) (2^32,88.406) (2^34,231.890)};
\addplot[color=black,]
    coordinates { (2^10, 0.0353872 ) (2^22, 177.3993)};  
\addplot[color=black,]
    coordinates { (2^10, 0.00600497) (2^38, 195.4196)};  
\addplot[color=black,]
    coordinates { (2^10, 0.01117887 ) (2^36, 186.7119)};  
\addplot[color=black,]
    coordinates { (2^10, 0.01663513 ) (2^34, 192.8559)};  
\addplot[color=black,]
    coordinates { (2^10, 0.01663513 ) (2^34, 192.8559)};  
\legend{$d=n=2$; naive}
\addlegendentry{$d=2$; $n=2$}
\addlegendentry{$d=4$; $n=2$}
\addlegendentry{$d=8$; $n=2$}
\end{axis}
\end{tikzpicture}
\caption{Timings of \textsf{Algorithm~3} implemented in \textsf{Maple~2020.2}. We compare $n=2$ and $d=2,4,8$ for various values of $N$.}
\label{fig:d}
\end{figure}

In Figure~\ref{fig:d} we show similar timings, however now for $n=2$ and $d=2,4,8$. The linear fits to the data are now given by $y=0.54x-12.7, y=0.54x-11.9$ and $y=0.56x-11.5$. Again, they describe the observations very well and the coefficients of the regressions are in line with the proven complexity. We observe that, as expected, the lines have slopes of roughly $1/2$ and are closer together than in the previous figure. The naive method has a slope of $1.02 \approx 1$.

%%%%%%%%%%%%%%%%%%%%%%%%%%%%%%%%%%%%%%%%%%%%%%%%%%%%%%%%%%%%%%%%%%%%

\section{Conclusion and future work}\label{sec:conclusion}

We have shown that selected terms of $q$-holonomic sequences can be computed
fast, {both in theory and in practice},
the key being the extension of classical algorithms in the holonomic (``$q=1$'') case.
We have demonstrated through {several} examples that this basic algorithmic
improvement has many other algorithmic implications, notably on the faster
evaluation of many families of polynomials and {on} the acceleration of
algorithms for $q$-difference equations.

{Here are some questions that should be investigated in the future.}
\begin{itemize}
\item[1.] (Counting points on $q$-curves) Counting efficiently points on
(hyper-)elliptic curves leads to questions like: for $a,b\in \Z$, compute the
coefficient of $x^{\frac{p-1}{2}}$ in $G_p(x)\coloneqq(x^3+ax+b)^{\frac{p-1}{2}}$ modulo
$p$, for one~\cite{BoGaSc07} or several~\cite{Harvey14} primes~$p$. A natural
extension is to ask the same with $G_p(x)$ replaced by
$\prod_{{k=1}}^{\frac{p-1}{2}} (q^{3k} x^3+aq^kx+b)$. This might have applications related to~\S\ref{sec:q-diff_curv}, or to counting points on
$q$-deformations~\cite{Scholze17}. 

\smallskip
\item[2.] (Computing $q$-deformed real numbers) Recently, 
	Morier-Genoud and Ovsienko~\cite{MoOv19} introduced $q$-analogues of 
	real numbers, see also~\cite{LeStQu15,MoOv20}. How fast can one compute (truncations or evaluations of) 
	quantized versions of numbers 
	like $e$ or $\pi$?

\smallskip
	\item[3.] (Evaluating more polynomials) Is it possible to evaluate 
	fast polynomials of the form $\sum_{\ell=0}^N x^{\ell^s}$, for $s\geq 3$, 
	and many others {that} escape the $q$-holonomic class?
	{\em E.g.}, \cite{BeDeLeSm20} presents a beautiful generalization of {\sf Algorithm 1}
	  to the fast evaluation of isogenies between elliptic curves, by
	  using \emph{elliptic resultants}, with applications to isogeny-based
	  cryptography.
	  
\smallskip
	\item[4.] (``Precise'' complexity of $q$-holonomic sequences) Can one prove non-trivial lower bounds, ideally matching the upper bounds, on examples treated in this article?
\end{itemize}			

%%%%%%%%%%%%%%%%%%%%%%%%%%%%%%%%%%%%%%%%%%%%%%%%%%%%%%%%%%%%%%%%%%%%

\noindent {\bf Acknowledgements.}
We would like to thank Luca De Feo for his initial question, who motivated this work, and for the very interesting subsequent discussions. 
Our warm thanks go to Lucia Di Vizio for valuable comments that led to \S\ref{sec:q-diff_curv}. Moreover, we thank her and Kilian Raschel for their careful reading of the first version of this manuscript. \\
The first author was supported in part by
\textcolor{magenta}{\href{https://specfun.inria.fr/chyzak/DeRerumNatura/}{DeRerumNatura}}
ANR-19-CE40-0018 and the second author by the \href{https://www.fwf.ac.at/}{Austrian
Science Fund} (FWF) P-31338.

%%%%%%%%%%%%%%%%%%%%%%%%%%%%%%%%%%%%%%%%%%%%%%%%%%%%%%%%%%%%%%%%%%%%

%\section*{References}

\def\gathen#1{{#1}}\def\cprime{$'$}
  \def\gathen#1{{#1}}\def\haesler#1{{#1}}\def\hoeij#1{{#1}}


\begin{thebibliography}{100}
\expandafter\ifx\csname url\endcsname\relax
  \def\url#1{\texttt{#1}}\fi
\expandafter\ifx\csname urlprefix\endcsname\relax\def\urlprefix{URL }\fi

\bibitem{AbBrPe95}
S.~Abramov, M.~Bronstein, M.~Petkov{\v s}ek, On polynomial solutions of linear
  operator equations, in: ISSAC'95, ACM Press, 1995, pp. 290--296.

\bibitem{Abramov95}
S.~A. Abramov, Rational solutions of linear difference and {$q$}-difference
  equations with polynomial coefficients, Programmirovanie~(6) (1995) 3--11.

\bibitem{Abramov02}
S.~A. Abramov, A direct algorithm to compute rational solutions of first order
  linear {$q$}-difference systems, Discrete Math. 246~(1-3) (2002) 3--12.
\newline\urlprefix\url{https://doi.org/10.1016/S0012-365X(01)00248-5}

\bibitem{APP98}
S.~A. Abramov, P.~Paule, M.~Petkov\v{s}ek, {$q$}-hypergeometric solutions of
  {$q$}-difference equations, Discrete Math. 180~(1-3) (1998) 3--22.
\newline\urlprefix\url{https://doi.org/10.1016/S0012-365X(97)00106-4}

\bibitem{AbZi96}
S.~A. Abramov, E.~V. Zima, D'{A}lembertian solutions of inhomogeneous linear
  equations (differential, difference, and some other), in: ISSAC'96, ACM
  Press, 1996, p. 232–240.
\newline\urlprefix\url{https://doi.org/10.1145/236869.237080}

\bibitem{AdBeDeJo17}
B.~Adamczewski, J.~P. Bell, E.~Delaygue, F.~Jouhet, Congruences modulo
  cyclotomic polynomials and algebraic independence for {$q$}-series, S\'{e}m.
  Lothar. Combin. 78B (2017) Art. 54, 12.

\bibitem{Adams1928}
C.~R. Adams, On the linear ordinary {$q$}-difference equation, Ann. of Math.
  (2) 30~(1-4) (1928/29) 195--205.
\newline\urlprefix\url{https://doi.org/10.2307/1968274}

\bibitem{Adams1931}
C.~R. Adams, Linear {$q$}-difference equations, Bull. AMS 37~(6) (1931)
  361--400.
\newline\urlprefix\url{https://doi.org/10.1090/S0002-9904-1931-05162-4}

\bibitem{AlCa65}
W.~A. Al-Salam, L.~Carlitz, Some orthogonal {$q$}-polynomials, Math. Nachr. 30
  (1965) 47--61.
\newline\urlprefix\url{https://doi.org/10.1002/mana.19650300105}

\bibitem{AMMP01}
M.~Aldaz, G.~Matera, J.~L. Monta\~{n}a, L.~M. Pardo, A new method to obtain
  lower bounds for polynomial evaluation, Theoret. Comput. Sci. 259~(1-2)
  (2001) 577--596.
\newline\urlprefix\url{https://doi.org/10.1016/S0304-3975(00)00149-3}

\bibitem{AlWi20}
J.~Alman, V.~Vassilevska~Williams, A refined laser method and faster matrix
  multiplication, arXiv preprint, 29 pages (2020).
\newline\urlprefix\url{https://arxiv.org/abs/2010.05846}

\bibitem{Andrews72}
G.~E. Andrews, Partition identities, Advances in Math. 9 (1972) 10--51.
\newline\urlprefix\url{https://doi.org/10.1016/0001-8708(72)90028-X}

\bibitem{Andrews74}
G.~E. Andrews, A general theory of identities of the {R}ogers-{R}amanujan type,
  Bull. AMS 80 (1974) 1033--1052.
\newline\urlprefix\url{https://doi.org/10.1090/S0002-9904-1974-13616-5}

\bibitem{Andrews1976}
G.~E. Andrews, The theory of partitions, Addison-Wesley Publishing Co.,
  Reading,, 1976, encyclopedia of Mathematics and its Applications, Vol. 2.

\bibitem{Andrews86}
G.~E. Andrews, The fifth and seventh order mock theta functions, Trans. Amer.
  Math. Soc. 293~(1) (1986) 113--134.
\newline\urlprefix\url{https://doi.org/10.2307/2000275}

\bibitem{Andrews10}
G.~E. Andrews, {$q$}-{C}atalan identities, in: The legacy of {A}lladi
  {R}amakrishnan in the mathematical sciences, Springer, 2010, pp. 183--190.
\newline\urlprefix\url{https://doi.org/10.1007/978-1-4419-6263-8_10}

\bibitem{ArZh20}
C.~E. Arreche, Y.~Zhang, Computing differential {G}alois groups of second-order
  linear $q$-difference equations, preprint, 2020.

\bibitem{Askey89}
R.~Askey, Continuous {$q$}-{H}ermite polynomials when {$q>1$}, in: {$q$}-series
  and partitions, vol.~18 of IMA Vol. Math. Appl., Springer, 1989, pp.
  151--158.
\newline\urlprefix\url{https://doi.org/10.1007/978-1-4684-0637-5_12}

\bibitem{BaGa96}
D.~Bar-Natan, S.~Garoufalidis, On the {M}elvin-{M}orton-{R}ozansky conjecture,
  Invent. Math. 125~(1) (1996) 103--133.
\newline\urlprefix\url{https://doi.org/10.1007/s002220050070}

\bibitem{BaClHa19}
M.~Barkatou, T.~Cluzeau, A.~El~Hajj, Simple forms and rational solutions of
  pseudo-linear systems, in: I{SSAC}'19, ACM Press, 2019, pp. 26--33.
\newline\urlprefix\url{https://doi.org/10.1145/3326229.3326246}

\bibitem{hakmem}
M.~Beeler, R.~Gosper, R.~Schroeppel, {HAKMEM}, Artificial Intelligence Memo
  No.~239, MIT, 1972,
  \textcolor{magenta}{\href{http://www.inwap.com/pdp10/hbaker/hakmem/algorithms.html}{http://www.inwap.com/pdp10/hbaker/hakmem/algorithms}}.

\bibitem{Bellman61}
R.~Bellman, A brief introduction to theta functions, Athena Series: Selected
  Topics in Mathematics, Holt, Rinehart and Winston, 1961.
\newline\urlprefix\url{https://doi.org/10.1017/s0025557200044491}

\bibitem{Bernstein08}
D.~J. Bernstein, Fast multiplication and its applications, in: Algorithmic
  number theory: lattices, number fields, curves and cryptography, vol.~44 of
  Math. Sci. Res. Inst. Publ., Cambridge Univ. Press, 2008, pp. 325--384.

\bibitem{BeDeLeSm20}
D.~J. Bernstein, L.~D. Feo, A.~Leroux, B.~Smith, Faster computation of
  isogenies of large prime degree, Cryptology ePrint Archive, Report 2020/341,
  \url{https://eprint.iacr.org/2020/341}, presented to
  \textcolor{magenta}{\href{https://www.math.auckland.ac.nz/~sgal018/ANTS/accepted.html}{ANTS-XIV}}
  (2020).

\bibitem{Bezivin91}
J.-P. B\'{e}zivin, Les \! suites \! {$q$}-r\'{e}currentes \! lin\'{e}aires,
  Compositio Math. 80~(3) (1991) 285--307.
\newline\urlprefix\url{http://www.numdam.org/item?id=CM_1991__80_3_285_0}

\bibitem{Bezivin92}
J.-P. B\'{e}zivin, Sur les \'{e}quations fonctionnelles aux
  {$q$}-diff\'{e}rences, Aequationes Math. 43~(2-3) (1992) 159--176.
\newline\urlprefix\url{https://doi.org/10.1007/BF01835698}

\bibitem{Bluestein70}
L.~I. Bluestein, A linear filtering approach to the computation of the discrete
  {F}ourier transform, IEEE Trans. Electroacoustics AU-18 (1970) 451--455.

\bibitem{BoKo99}
H.~B\"{o}ing, W.~Koepf, Algorithms for {$q$}-hypergeometric summation in
  computer algebra, J. Symbolic Comput. 28~(6) (1999) 777--799.
\newline\urlprefix\url{https://doi.org/10.1006/jsco.1998.0339}

\bibitem{BoCo76}
A.~Borodin, S.~Cook, On the number of additions to compute specific
  polynomials, SIAM J. Comput. 5~(1) (1976) 146--157.
\newline\urlprefix\url{https://doi.org/10.1137/0205013}

\bibitem{Borwein88}
P.~B. Borwein, Pad\'{e} approximants for the {$q$}-elementary functions,
  Constr. Approx. 4~(4) (1988) 391--402.
\newline\urlprefix\url{https://doi.org/10.1007/BF02075469}

\bibitem{Bostan20}
A.~Bostan, {Computing the $N$-th Term of a $q$-Holonomic Sequence}, in:
  {ISSAC'20}, {ACM Press}, 2020, pp. 46--53.

\bibitem{BoCaSc14}
A.~Bostan, X.~Caruso, E.~Schost, A fast algorithm for computing the
  characteristic polynomial of the $p$-curvature, in: ISSAC'14, ACM Press,
  2014, pp. 59--66.
\newline\urlprefix\url{https://doi.org/10.1145/2608628.2608650}

\bibitem{BoCaSc15}
A.~Bostan, X.~Caruso, E.~Schost, A fast algorithm for computing the
  {$p$}-curvature, in: I{SSAC}'15, ACM, New York, 2015, pp. 69--76.

\bibitem{BoCaSc16}
A.~Bostan, X.~Caruso, E.~Schost, Computation of the similarity class of the
  {$p$}-curvature, in: ISSAC'16, ACM Press, 2016, pp. 111--118.

\bibitem{BCCS06}
A.~Bostan, F.~Chyzak, T.~Cluzeau, B.~Salvy, Low complexity algorithms for
  linear recurrences, in: I{SSAC}'06, ACM Press, 2006, pp. 31--38.
\newline\urlprefix\url{https://doi.org/10.1145/1145768.1145781}

\bibitem{AECF}
A.~Bostan, F.~Chyzak, M.~Giusti, R.~Lebreton, G.~Lecerf, B.~Salvy, E.~Schost,
  Algorithmes Efficaces en Calcul Formel, Palaiseau, 2017, 686 pages, in
  French. Printed by CreateSpace. Also available in electronic form.
\newline\urlprefix\url{https://hal.archives-ouvertes.fr/AECF/}

\bibitem{BoClSa05}
A.~Bostan, T.~Cluzeau, B.~Salvy, Fast algorithms for polynomial solutions of
  linear differential equations, in: I{SSAC}'05, ACM Press, 2005, pp. 45--52.
\newline\urlprefix\url{https://doi.org/10.1145/1073884.1073893}

\bibitem{BoGaSc07}
A.~Bostan, P.~Gaudry, E.~Schost, Linear recurrences with polynomial
  coefficients and application to integer factorization and {C}artier-{M}anin
  operator, SIAM J. Comput. 36~(6) (2007) 1777--1806.
\newline\urlprefix\url{https://doi.org/10.1137/S0097539704443793}

\bibitem{BoKa09}
A.~Bostan, M.~Kauers, Automatic classification of restricted lattice walks, in:
  {FPSAC} 2009, Discrete Math. Theor. Comput. Sci. Proc., AK, Assoc. Discrete
  Math. Theor. Comput. Sci., Nancy, 2009, pp. 201--215.

\bibitem{BoLeSc03}
A.~Bostan, G.~Lecerf, {\'E}.~Schost, Tellegen's principle into practice, in:
  Proceedings of ISSAC'03, ACM Press, 2003, pp. 37--44.

\bibitem{BoMo21}
A.~Bostan, R.~Mori, {A Simple and Fast Algorithm for Computing the $N$-th Term
  of a Linearly Recurrent Sequence}, in: SOSA'21 (Symposium on Simplicity in
  Algorithms), SIAM, 2021.
\newline\urlprefix\url{https://specfun.inria.fr/bostan/BoMo21.pdf}

\bibitem{BoSc05}
A.~Bostan, E.~Schost, Polynomial evaluation and interpolation on special sets
  of points, J. Complexity 21~(4) (2005) 420--446.
\newline\urlprefix\url{https://doi.org/10.1016/j.jco.2004.09.009}

\bibitem{BoSc09}
A.~Bostan, E.~Schost, Fast algorithms for differential equations in positive
  characteristic, in: ISSAC'09, ACM Press, 2009, pp. 47--54.
\newline\urlprefix\url{https://doi.org/10.1145/1576702.1576712}

\bibitem{BousquetMelou92}
M.~Bousquet-M\'{e}lou, Convex polyominoes and algebraic languages, J. Phys. A
  25~(7) (1992) 1935--1944.
\newline\urlprefix\url{http://stacks.iop.org/0305-4470/25/1935}

\bibitem{BCS97}
P.~B\"{u}rgisser, M.~Clausen, M.~A. Shokrollahi, Algebraic complexity theory,
  vol. 315 of Grundlehren der Mathematischen Wissenschaften, Springer, 1997.
\newline\urlprefix\url{https://doi.org/10.1007/978-3-662-03338-8}

\bibitem{CaKa91}
D.~G. Cantor, E.~Kaltofen, On fast multiplication of polynomials over arbitrary
  algebras, Acta Inform. 28~(7) (1991) 693--701.

\bibitem{Carmichael1912}
R.~D. Carmichael, The {G}eneral {T}heory of {L}inear {$q$}-{D}ifference
  {E}quations, Amer. J. Math. 34~(2) (1912) 147--168.
\newline\urlprefix\url{https://doi.org/10.2307/2369887}

\bibitem{Cartier92}
P.~Cartier, D\'{e}monstration ``automatique'' d'identit\'{e}s et fonctions
  hyperg\'{e}om\'{e}triques (d'apr\`es {D}. {Z}eilberger), Ast\'{e}risque~(206)
  (1992) Exp. No. 746, 3, 41--91, s\'{e}minaire Bourbaki, Vol. 1991/92.

\bibitem{ChuChu88}
D.~V. Chudnovsky, G.~V. Chudnovsky, Approximations and complex multiplication
  according to {R}amanujan, in: Ramanujan revisited ({U}rbana-{C}hampaign,
  {I}ll., 1987), Academic Press, Boston, MA, 1988, pp. 375--472.

\bibitem{Chyzak98}
F.~Chyzak, Gr\"{o}bner bases, symbolic summation and symbolic integration, in:
  Gr\"{o}bner bases and applications ({L}inz, 1998), vol. 251 of London Math.
  Soc. Lecture Note Ser., Cambridge Univ. Press, 1998, pp. 32--60.

\bibitem{Chyzak00}
F.~Chyzak, An extension of {Z}eilberger's fast algorithm to general holonomic
  functions, Discrete Math. 217~(1-3) (2000) 115--134.
\newline\urlprefix\url{https://doi.org/10.1016/S0012-365X(99)00259-9}

\bibitem{ChDuMaMiSa16}
F.~Chyzak, P.~Dumas, H.~Le, J.~Martin, M.~M., B.~Salvy, Taming apparent
  singularities via {O}re closure, manuscript, 2016.

\bibitem{CoGeHa14}
E.~Costa, R.~Gerbicz, D.~Harvey, A search for {W}ilson primes, Math. Comp.
  83~(290) (2014) 3071--3091.
\newline\urlprefix\url{https://doi.org/10.1090/S0025-5718-2014-02800-7}

\bibitem{DeGa20}
R.~Detcherry, S.~Garoufalidis, A diagrammatic approach to the {AJ} conjecture,
  Math. Ann. 378~(1-2) (2020) 447--484.
\newline\urlprefix\url{https://doi.org/10.1007/s00208-020-02028-y}

\bibitem{Vizio02}
L.~Di~Vizio, Arithmetic theory of {$q$}-difference equations: the
  {$q$}-analogue of {G}rothendieck-{K}atz's \! conjecture on {$p$}-curvatures,
  Invent. Math. 150~(3) (2002) 517--578.
\newline\urlprefix\url{https://doi.org/10.1007/s00222-002-0241-z}

\bibitem{DiHa20}
L.~Di~Vizio, C.~Hardouin, Intrinsic approach to {G}alois theory of
  $q$-difference equations, with the preface to {P}art {IV} by {A}nne
  {G}ranier, Mem. Amer. Math. Soc.~77 pages. To appear.
\newline\urlprefix\url{http://arxiv.org/abs/1002.4839}

\bibitem{DRSZ03}
L.~Di~Vizio, J.-P. Ramis, J.~Sauloy, C.~Zhang, \'{E}quations aux
  {$q$}-diff\'{e}rences, Gaz. Math.~(96) (2003) 20--49.

\bibitem{EkGe15}
T.~Ekedahl, G.~van~der Geer, Cycle classes on the moduli of {K}3 surfaces in
  positive characteristic, Selecta Math. (N.S.) 21~(1) (2015) 245--291.
\newline\urlprefix\url{https://doi.org/10.1007/s00029-014-0156-8}

\bibitem{EZ96}
S.~B. Ekhad, D.~Zeilberger, The number of solutions of {$X^2=0$} in triangular
  matrices over {${\rm GF}(q)$}, Electron. J. Combin. 3~(1) (1996) Research
  Paper 2, approx. 2.
\newline\urlprefix\url{http://www.combinatorics.org/Volume_3/Abstracts/v3i1r2.html}

\bibitem{Ernst12}
T.~Ernst, A comprehensive treatment of {$q$}-calculus, Birkh\"{a}user/Springer,
  2012.
\newline\urlprefix\url{https://doi.org/10.1007/978-3-0348-0431-8}

\bibitem{Fiduccia85}
C.~M. Fiduccia, An efficient formula for linear recurrences, SIAM J. Comput.
  14~(1) (1985) 106--112.
\newline\urlprefix\url{https://doi.org/10.1137/0214007}

\bibitem{Furer09}
M.~F\"{u}rer, Faster integer multiplication, SIAM J. Comput. 39~(3) (2009)
  979--1005.
\newline\urlprefix\url{https://doi.org/10.1137/070711761}

\bibitem{FuHo85}
J.~F\"{u}rlinger, J.~Hofbauer, {$q$}-{C}atalan numbers, J. Combin. Theory Ser.
  A 40~(2) (1985) 248--264.
\newline\urlprefix\url{https://doi.org/10.1016/0097-3165(85)90089-5}

\bibitem{LeGall14}
F.~L. Gall, Powers of tensors and fast matrix multiplication, in: ISSAC'14,
  2014, pp. 296--303.

\bibitem{Garoufalidis04}
S.~Garoufalidis, On the characteristic and deformation varieties of a knot, in:
  Proceedings of the {C}asson {F}est, vol.~7 of Geom. Topol. Monogr., Geom.
  Topol. Publ., Coventry, 2004, pp. 291--309.
\newline\urlprefix\url{https://doi.org/10.2140/gtm.2004.7.291}

\bibitem{Garoufalidis11}
S.~Garoufalidis, The degree of a {$q$}-holonomic sequence is a quadratic
  quasi-polynomial, Electron. J. Combin. 18~(2) (2011) Paper 4, 23.
\newline\urlprefix\url{https://doi.org/10.37236/2000}

\bibitem{Garoufalidis18}
S.~Garoufalidis, Quantum knot invariants, Res. Math. Sci. 5~(1) (2018) Paper
  No. 11, 17.
\newline\urlprefix\url{https://doi.org/10.1007/s40687-018-0127-3}

\bibitem{GaKo12}
S.~Garoufalidis, C.~Koutschan, Twisting {$q$}-holonomic sequences by complex
  roots of unity, in: I{SSAC}'12, ACM Press, 2012, pp. 179--186.
\newline\urlprefix\url{https://doi.org/10.1145/2442829.2442857}

\bibitem{GaKo13}
S.~Garoufalidis, C.~Koutschan, Irreducibility of {$q$}-difference operators and
  the knot {$7_4$}, Algebr. Geom. Topol. 13~(6) (2013) 3261--3286.
\newline\urlprefix\url{https://doi.org/10.2140/agt.2013.13.3261}

\bibitem{GaLaLe18}
S.~Garoufalidis, A.~D. Lauda, T.~T.~Q. L\^{e}, The colored {HOMFLYPT} function
  is {$q$}-holonomic, Duke Math. J. 167~(3) (2018) 397--447.
\newline\urlprefix\url{https://doi.org/10.1215/00127094-2017-0030}

\bibitem{GaLe05}
S.~Garoufalidis, T.~T.~Q. L\^{e}, The colored {J}ones function is
  {$q$}-holonomic, Geom. Topol. 9 (2005) 1253--1293.
\newline\urlprefix\url{https://doi.org/10.2140/gt.2005.9.1253}

\bibitem{GaLe16}
S.~Garoufalidis, T.~T.~Q. L\^{e}, A survey of {$q$}-holonomic functions,
  Enseign. Math. 62~(3-4) (2016) 501--525.
\newline\urlprefix\url{https://doi.org/10.4171/LEM/62-3/4-7}

\bibitem{Garvan19}
F.~G. Garvan, New fifth and seventh order mock theta function identities, Ann.
  Comb. 23~(3-4) (2019) 765--783.
\newline\urlprefix\url{https://doi.org/10.1007/s00026-019-00438-7}

\bibitem{GaGe13}
J.~\gathen{von zur} Gathen, J.~Gerhard, Modern computer algebra, 3rd ed.,
  Cambridge Univ. Press, 2013.
\newline\urlprefix\url{https://doi.org/10.1017/CBO9781139856065}

\bibitem{Gauss1808}
C.~F. Gauss, Summatio quarundam serierum singularium, Opera, Vol. 2,
  G\"ottingen: Gess. d. Wiss. (1863).

\bibitem{Gessel80}
I.~Gessel, A noncommutative generalization and {$q$}-analog of the {L}agrange
  inversion formula, Trans. Amer. Math. Soc. 257~(2) (1980) 455--482.
\newline\urlprefix\url{https://doi.org/10.2307/1998307}

\bibitem{Hahn50}
W.~Hahn, \"{U}ber die h\"{o}heren {H}eineschen {R}eihen und eine einheitliche
  {T}heorie der sogenannten speziellen {F}unktionen, Math. Nachr. 3 (1950)
  257--294.
\newline\urlprefix\url{https://doi.org/10.1002/mana.19490030502}

\bibitem{HQZ04}
G.~Hanrot, M.~Quercia, P.~Zimmermann, The middle product algorithm. {I}, Appl.
  Algebra Engrg. Comm. Comput. 14~(6) (2004) 415--438.
\newline\urlprefix\url{https://doi.org/10.1007/s00200-003-0144-2}

\bibitem{Harvey14}
D.~Harvey, Counting points on hyperelliptic curves in average polynomial time,
  Ann. of Math. (2) 179~(2) (2014) 783--803.
\newline\urlprefix\url{https://doi.org/10.4007/annals.2014.179.2.7}

\bibitem{HaHo20}
D.~Harvey, J.~van~der Hoeven, Integer multiplication in time ${O}(n \log n)$,
  Ann. of Math.{} To appear.

\bibitem{HaHoLe16}
D.~Harvey, J.~van~der Hoeven, G.~Lecerf, Even faster integer multiplication, J.
  Complexity 36 (2016) 1--30.
\newline\urlprefix\url{https://doi.org/10.1016/j.jco.2016.03.001}

\bibitem{Heine1847}
E.~Heine, Untersuchungen \"uber die {R}eihe
  $1+\frac{(1-q^{\alpha})(1-q^{\beta})}{(1-q)(1-q^{\gamma})} x +
  \frac{(1-q^{\alpha})(1-q^{\alpha +1})(1-q^{\beta})
  (1-q^{\beta+1})}{(1-q)(1-q^2)(1-q^{\gamma})(1-q^{\gamma +1})} x^2 + .. $, J.
  reine angew. Math. 34 (1847) 285--328.

\bibitem{HeSi80}
J.~Heintz, M.~Sieveking, Lower bounds for polynomials with algebraic
  coefficients, Theoret. Comput. Sci. 11~(3) (1980) 321--330.
\newline\urlprefix\url{https://doi.org/10.1016/0304-3975(80)90019-5}

\bibitem{Hendriks97}
P.~A. Hendriks, An algorithm for computing a standard form for second-order
  linear {$q$}-difference equations, J. Pure Appl. Algebra 117/118 (1997)
  331--352.
\newline\urlprefix\url{https://doi.org/10.1016/S0022-4049(97)00017-0}

\bibitem{Hua00}
J.~Hua, Counting representations of quivers over finite fields, J. Algebra
  226~(2) (2000) 1011--1033.
\newline\urlprefix\url{https://doi.org/10.1006/jabr.1999.8220}

\bibitem{Ismail01}
M.~E.~H. Ismail, Lectures on {$q$}-orthogonal polynomials, in: Special
  functions 2000: current perspective and future directions ({T}empe, {AZ}),
  vol.~30 of NATO Sci. Ser. II Math. Phys. Chem., Kluwer Acad. Publ., 2001, pp.
  179--219.
\newline\urlprefix\url{https://doi.org/10.1007/978-94-010-0818-1_8}

\bibitem{Jackson1909}
F.~H. Jackson, On $q$-functions and a certain difference operator, Trans. Roy.
  Soc. Edin. 46~(11) (1909) 253--281.
\newline\urlprefix\url{https://doi.org/10.1017/S0080456800002751}

\bibitem{Jackson1910b}
F.~H. Jackson, On \!\! {$q$}-definite \!\! integrals, \!\! Quar. \!\! J. \!\!
  Pure \! Appl. \!\! Math. 41 (1910) 193--203.

\bibitem{Jackson1910}
F.~H. Jackson, {$q$}-{D}ifference {E}quations, Amer. J. Math. 32~(4) (1910)
  305--314.
\newline\urlprefix\url{https://doi.org/10.2307/2370183}

\bibitem{Katz72}
N.~M. Katz, Algebraic solutions of differential equations ({$p$}-curvature and
  the {H}odge filtration), Invent. Math. 18 (1972) 1--118.
\newline\urlprefix\url{https://doi.org/10.1007/BF01389714}

\bibitem{KaKo09}
M.~Kauers, C.~Koutschan, A {M}athematica package for {$q$}-holonomic sequences
  and power series, Ramanujan J. 19~(2) (2009) 137--150.
\newline\urlprefix\url{https://doi.org/10.1007/s11139-008-9132-2}

\bibitem{Khmelnov00}
D.~E. Khmel\cprime~nov, Improved algorithms for solving difference and
  {$q$}-difference equations, Programmirovanie~(2) (2000) 70--78.
\newline\urlprefix\url{https://doi.org/10.1007/BF02759198}

\bibitem{Kirillov95}
A.~A. Kirillov, On the number of solutions of the equation {$X^2=0$} in
  triangular matrices over a finite field, Funktsional. Anal. i Prilozhen.
  29~(1) (1995) 82--87.
\newline\urlprefix\url{https://doi.org/10.1007/BF01077044}

\bibitem{KiMe97}
A.~A. Kirillov, A.~Melnikov, On a remarkable sequence of polynomials, in:
  Alg\`ebre non commutative, groupes quantiques et invariants ({R}eims, 1995),
  vol.~2 of S\'{e}min. Congr., Soc. Math. France, Paris, 1997, pp. 35--42.

\bibitem{KoKeSw10}
R.~Koekoek, P.~A. Lesky, R.~F. Swarttouw, Hypergeometric orthogonal polynomials
  and their {$q$}-analogues, Monographs in Mathematics, Springer, 2010.
\newline\urlprefix\url{https://doi.org/10.1007/978-3-642-05014-5}

\bibitem{KPS07}
W.~Koepf, P.~M. Rajkovi\'{c}, S.~D. Marinkovi\'{c}, Properties of
  {$q$}-holonomic functions, J. Difference Equ. Appl. 13~(7) (2007) 621--638.
\newline\urlprefix\url{https://doi.org/10.1080/10236190701264925}

\bibitem{Koornwinder93}
T.~H. Koornwinder, On {Z}eilberger's algorithm and its {$q$}-analogue, J.
  Comput. Appl. Math. 48~(1-2) (1993) 91--111.
\newline\urlprefix\url{https://doi.org/10.1016/0377-0427(93)90317-5}

\bibitem{Koutschan10}
C.~Koutschan, A fast approach to creative telescoping, Math. Comput. Sci.
  4~(2-3) (2010) 259--266.
\newline\urlprefix\url{https://doi.org/10.1007/s11786-010-0055-0}

\bibitem{KoZh18}
C.~Koutschan, Y.~Zhang, Desingularization in the {$q$}-{W}eyl algebra, Adv. in
  Appl. Math. 97 (2018) 80--101.
\newline\urlprefix\url{https://doi.org/10.1016/j.aam.2018.02.005}

\bibitem{Labrande18}
H.~Labrande, Computing {J}acobi's theta in quasi-linear time, Math. Comp.
  87~(311) (2018) 1479--1508.
\newline\urlprefix\url{https://doi.org/10.1090/mcom/3245}

\bibitem{LTh16}
H.~Labrande, E.~Thom\'{e}, Computing theta functions in quasi-linear time in
  genus two and above, LMS J. Comput. Math. 19~(suppl. A) (2016) 163--177.
\newline\urlprefix\url{https://doi.org/10.1112/S1461157016000309}

\bibitem{LeCaine1943}
J.~Le~Caine, The linear {$q$}-difference equation of the second order, Amer. J.
  Math. 65 (1943) 585--600.
\newline\urlprefix\url{https://doi.org/10.2307/2371867}

\bibitem{LeStQu15}
B.~Le~Stum, A.~Quir\'{o}s, On quantum integers and rationals, in: Trends in
  number theory, vol. 649 of Contemp. Math., Amer. Math. Soc., Providence, RI,
  2015, pp. 107--130.
\newline\urlprefix\url{https://doi.org/10.1090/conm/649/13022}

\bibitem{Lipton78}
R.~J. Lipton, Polynomials with {$0-1$} coefficients that are hard to evaluate,
  SIAM J. Comput. 7~(1) (1978) 61--69.
\newline\urlprefix\url{https://doi.org/10.1137/0207004}

\bibitem{Liu17}
J.-C. Liu, Some finite generalizations of {E}uler's pentagonal number theorem,
  Czechoslovak Math. J. 67(142)~(2) (2017) 525--531.
\newline\urlprefix\url{https://doi.org/10.21136/CMJ.2017.0063-16}

\bibitem{Mason1915}
T.~E. Mason, On {P}roperties of the {S}olutions of {L}inear {$q$}-{D}ifference
  {E}quations with {E}ntire {F}unction {C}oefficients, Amer. J. Math. 37~(4)
  (1915) 439--444.
\newline\urlprefix\url{https://doi.org/10.2307/2370216}

\bibitem{MiBr66}
J.~C.~P. Miller, D.~J.~S. Brown, An algorithm for evaluation of remote terms in
  a linear recurrence sequence, Comput. J. 9 (1966) 188--190.

\bibitem{MoOv19}
S.~Morier-Genoud, V.~Ovsienko, On $q$-deformed real numbers, Exp. Math. (2019)
  1--9.{} To appear.
\newline\urlprefix\url{https://doi.org/10.1080/10586458.2019.1671922}

\bibitem{MoOv20}
S.~Morier-Genoud, V.~Ovsienko, {$q$}-deformed rationals and {$q$}-continued
  fractions, Forum Math. Sigma 8 (2020) Paper No. e13, 55.
\newline\urlprefix\url{https://doi.org/10.1017/fms.2020.9}

\bibitem{NognengSchost18}
D.~Nogneng, E.~Schost, On the evaluation of some sparse polynomials, Math.
  Comp. 87~(310) (2018) 893--904.
\newline\urlprefix\url{https://doi.org/10.1090/mcom/3231}

\bibitem{Osgood1971}
C.~F. Osgood, On the diophantine approximation of values of functions
  satisfying certain linear {$q$}-difference equations, J. Number Theory 3
  (1971) 159--177.
\newline\urlprefix\url{https://doi.org/10.1016/0022-314X(71)90033-3}

\bibitem{Ostrowski54}
A.~Ostrowski, On two problems in abstract algebra connected with {H}orner's
  rule, in: Studies in mathematics and mechanics presented to {R}ichard von
  {M}ises, Academic Press Inc., 1954, pp. 40--48.

\bibitem{Pak06}
I.~Pak, Partition bijections, a survey, Ramanujan J. 12~(1) (2006) 5--75.
\newline\urlprefix\url{https://doi.org/10.1007/s11139-006-9576-1}

\bibitem{Pan66}
V.~Y. Pan, Methods of computing values of polynomials, Russian Mathematical
  Surveys 21~(1) (1966) 105--136.

\bibitem{PaSt71}
M.~{Paterson}, L.~{Stockmeyer}, Bounds on the evaluation time for rational
  polynomial, in: 12th Annual Symposium on Switching and Automata Theory
  ({SWAT} 1971), 1971, pp. 140--143.
\newline\urlprefix\url{https://doi.org/10.1109/SWAT.1971.5}

\bibitem{PaSc73}
M.~S. Paterson, L.~J. Stockmeyer, On the number of nonscalar multiplications
  necessary to evaluate polynomials, SIAM J. Comput. 2 (1973) 60--66.
\newline\urlprefix\url{https://doi.org/10.1137/0202007}

\bibitem{PaRa18}
P.~Paule, S.~Radu, Rogers-{R}amanujan functions, modular functions, and
  computer algebra, in: Advances in computer algebra, vol. 226 of Springer
  Proc. Math. Stat., Springer, Cham, 2018, pp. 229--280.
\newline\urlprefix\url{https://doi.org/10.1007/978-3-319-73232-9_10}

\bibitem{PauleRiese97}
P.~Paule, A.~Riese, A {M}athematica {$q$}-analogue of {Z}eilberger's algorithm
  based on an algebraically motivated approach to {$q$}-hypergeometric
  telescoping, in: Special functions, {$q$}-series and related topics, vol.~14
  of Fields Inst. Commun., Amer. Math. Soc., 1997, pp. 179--210.

\bibitem{PWZ96}
M.~Petkov\v{s}ek, H.~S. Wilf, D.~Zeilberger, {$A=B$}, A K Peters, 1996.

\bibitem{Pollard}
J.~M. Pollard, Theorems on factorization and primality testing, Proc. Cambridge
  Philos. Soc. 76 (1974) 521--528.
\newline\urlprefix\url{https://doi.org/10.1017/s0305004100049252}

\bibitem{RaScRa69}
L.~R. Rabiner, R.~W. Schafer, C.~M. Rader, The chirp {$z$}-transform algorithm
  and its application, Bell System Tech. J. 48 (1969) 1249--1292.
\newline\urlprefix\url{https://doi.org/10.1002/j.1538-7305.1969.tb04268.x}

\bibitem{Riese03}
A.~Riese, q{M}ulti{S}um---a package for proving {$q$}-hypergeometric multiple
  summation identities, J. Symbolic Comput. 35~(3) (2003) 349--376.
\newline\urlprefix\url{https://doi.org/10.1016/S0747-7171(02)00138-4}

\bibitem{RoRa19}
L.~Rogers, S.~Ramanujan, Proof of certain identities in combinatory analysis,
  {P}roc. {C}ambridge {P}hilos. {S}oc. 19 (1919) 211--214.

\bibitem{Rogers1893}
L.~J. Rogers, Second {M}emoir on the {E}xpansion of certain {I}nfinite
  {P}roducts, Proc. Lond. Math. Soc. 25 (1893/94) 318--343.
\newline\urlprefix\url{https://doi.org/10.1112/plms/s1-25.1.318}

\bibitem{Rothe1793}
H.~A. Rothe, Formulae de serierum reversione demonstratio universalis signis
  localibus combinatorico-analyticorum vicariis exhibita, Leipzig (1793).

\bibitem{Sabbah93}
C.~Sabbah, Syst\`emes holonomes d'\'{e}quations aux {$q$}-diff\'{e}rences, in:
  {$D$}-modules and microlocal geometry ({L}isbon, 1990), de Gruyter, 1993, pp.
  125--147.

\bibitem{Schnorr78}
C.-P. Schnorr, Improved lower bounds on the number of multiplications /
  divisions which are necessary to evaluate polynomials, Theoret. Comput. Sci.
  7~(3) (1978) 251--261.
\newline\urlprefix\url{https://doi.org/10.1016/0304-3975(78)90016-6}

\bibitem{Scholze17}
P.~Scholze, Canonical {$q$}-deformations in arithmetic geometry, Ann. Fac. Sci.
  Toulouse Math. (6) 26~(5) (2017) 1163--1192.
\newline\urlprefix\url{https://doi.org/10.5802/afst.1563}

\bibitem{Schoenhage77}
A.~Sch{\"o}nhage, Schnelle {M}ultiplikation von {P}olynomen \"uber {K}\"orpern
  der {C}harakteristik 2, Acta Informatica 7 (1977) 395--398.

\bibitem{ScSt71}
A.~Sch\"{o}nhage, V.~Strassen, Schnelle {M}ultiplikation grosser {Z}ahlen,
  Computing (Arch. Elektron. Rechnen) 7 (1971) 281--292.
\newline\urlprefix\url{https://doi.org/10.1007/bf02242355}

\bibitem{Shanks51}
D.~Shanks, A short proof of an identity of {E}uler, Proc. AMS 2 (1951)
  747--749.
\newline\urlprefix\url{https://doi.org/10.2307/2032076}

\bibitem{SprengerKoepf12}
T.~Sprenger, W.~Koepf, Algorithmic determination of {$q$}-power series for
  {$q$}-holonomic functions, J. Symbolic Comput. 47~(5) (2012) 519--535.
\newline\urlprefix\url{https://doi.org/10.1016/j.jsc.2011.12.004}

\bibitem{Strassen74}
V.~Strassen, Polynomials with rational coefficients which are hard to compute,
  SIAM J. Comput. 3 (1974) 128--149.
\newline\urlprefix\url{https://doi.org/10.1137/0203010}

\bibitem{Strassen77}
V.~Strassen, Einige {R}esultate \"{u}ber {B}erechnungskomplexit\"{a}t, Jber.
  Deutsch. Math.-Verein. 78~(1) (1976/77) 1--8.

\bibitem{Tanner1895}
H.~W.~L. Tanner, On the {E}numeration of {G}roups of {T}otitives, Proc. Lond.
  Math. Soc. 27 (1895/96) 329--352.
\newline\urlprefix\url{https://doi.org/10.1112/plms/s1-27.1.329}

\bibitem{TCH12}
T.~Tao, E.~Croot, III, H.~Helfgott, Deterministic methods to find primes, Math.
  Comp. 81~(278) (2012) 1233--1246.
\newline\urlprefix\url{https://doi.org/10.1090/S0025-5718-2011-02542-1}

\bibitem{Trjitzinsky1933}
W.~J. Trjitzinsky, Analytic theory of linear {$q$}-difference equations, Acta
  Math. 61~(1) (1933) 1--38.
\newline\urlprefix\url{https://doi.org/10.1007/BF02547785}

\bibitem{Tsai00}
H.~Tsai, Weyl closure of a linear differential operator, vol.~29, 2000, pp.
  747--775, {S}ymbolic computation in algebra, analysis, and geometry
  (Berkeley, CA, 1998).
\newline\urlprefix\url{https://doi.org/10.1006/jsco.1999.0400}

\bibitem{PuSi03}
M.~van~der Put, M.~F. Singer, Galois theory of linear differential equations,
  vol. 328 of Grundlehren der Mathematischen Wissenschaften, Springer, 2003.
\newline\urlprefix\url{https://doi.org/10.1007/978-3-642-55750-7}

\bibitem{WZ92}
H.~S. Wilf, D.~Zeilberger, An algorithmic proof theory for hypergeometric
  (ordinary $\&$ {$q$}) multisum/integral identities, Invent. Math. 108~(3)
  (1992) 575--633.
\newline\urlprefix\url{https://doi.org/10.1007/BF02100618}

\bibitem{Yang90}
K.-W. Yang, On the product {$\prod_{n\geq 1}(1+q^nx+q^{2n}x^2)$}, J. Austral.
  Math. Soc. Ser.~A 48~(1) (1990) 148--151.

\bibitem{Yip18}
M.~Yip, Rook placements and {J}ordan forms of upper-triangular nilpotent
  matrices, Electron. J. Combin. 25~(1) (2018) Paper 1.68, 25.

\bibitem{Zagier08}
D.~Zagier, Elliptic modular forms and their applications, in: The 1-2-3 of
  modular forms, Universitext, Springer, 2008, pp. 1--103.
\newline\urlprefix\url{https://doi.org/10.1007/978-3-540-74119-0_1}

\bibitem{Zeilberger90}
D.~Zeilberger, A holonomic systems approach to special functions identities, J.
  Comput. Appl. Math. 32~(3) (1990) 321--368.
\newline\urlprefix\url{https://doi.org/10.1016/0377-0427(90)90042-X}

\end{thebibliography}
\end{document}